\documentclass[letterpaper,twocolumn,10pt]{article}

\let\isExtended\undefined

\usepackage{usenix-2020-09}

\usepackage{graphicx}

\usepackage{amsmath}

\usepackage{amssymb}
\usepackage{amsthm}

\usepackage{xspace} 
\usepackage{url}

\usepackage{caption}
\usepackage{subcaption}
\usepackage{outlines} 

\usepackage{enumitem}
\setlist{nolistsep}

\usepackage[subtle]{savetrees}

\usepackage{titling}

\usepackage{appendix}

\newcommand{\squishlist}{
   \begin{list}{$\bullet$}
    { \setlength{\itemsep}{0pt}      \setlength{\parsep}{3pt}
      \setlength{\topsep}{3pt}       \setlength{\partopsep}{0pt}
      \setlength{\leftmargin}{1.0em} \setlength{\labelwidth}{1em}
      \setlength{\labelsep}{0.5em} } }
\newcommand{\squishend}{
    \end{list}  }

\numberwithin{equation}{section}

\usepackage[linesnumbered,ruled,vlined]{algorithm2e}

\usepackage{adjustbox}

\usepackage{listings}

\usepackage{multirow}
\usepackage{colortbl}
\usepackage{makecell}
\usepackage{hhline}
\usepackage{array}
\usepackage{amsmath}
\usepackage{mathtools}

\newcommand{\allnotes}[1]{\textit{#1}}

\newcommand{\fc}[1]{\allnotes{{\it\color{purple}[FC: #1]}}}
\newcommand{\jj}[1]{\allnotes{{\it\color{green}[JJ: #1]}}}

\newcommand{\neil}[1]{}

\newcommand{\SecNS}[1]{\S\ref{sec:#1}\xspace}
\newcommand{\Sec}[1]{\xspace\S\ref{sec:#1}\xspace}
\newcommand{\App}[1]{Appendix~\ref{app:#1}}
\newcommand{\Eqtn}[1]{Equation~\ref{eq:#1}}
\newcommand{\Tab}[1]{Table~\ref{tab:#1}}

\newcommand{\ie}{i.e.,\xspace}
\newcommand{\eg}{e.g.,\xspace}

\newcommand{\Para}[1]{\smallskip\noindent\textbf{#1}}

\newcommand{\nop}[1]{}
\newcommand{\Fig}[1]{Fig.~\ref{fig:#1}}
\newcommand{\system}{$\textsc{Privid}$\xspace}
\newcommand{\System}{$\textsc{Privid}$\xspace}

\providecommand{\abs}[1]{\lvert#1\rvert} 

\DeclareMathAlphabet{\mathpzc}{OT1}{pzc}{m}{it}

\theoremstyle{definition}
\newtheorem{defn}{\normalfont\bfseries\scshape Definition}[section]
\newtheorem{theorem}{Theorem}[section]
  \newtheorem{lemma}[theorem]{Lemma}

\newcommand{\tightcaption}[1]{\vspace{-15pt}\caption{\normalfont {\small #1}}
\vspace{-13pt}
}

\newcommand{\auburn}{\texttt{campus}}
\newcommand{\hampton}{\texttt{highway}}
\newcommand{\shibuya}{\texttt{urban}}

\newcommand{\pkbounded}{$(\rho, K)$-bounded}
\newcommand{\pkneighboring}{$(\rho, K)$-neighboring}
\newcommand{\pkeprivacylong}{$(\rho, K, \epsilon)$-event-duration privacy}
\newcommand{\pkeprivacy}{$(\rho, K, \epsilon)$-privacy\xspace}
\newcommand{\pkeprivate}{$(\rho, K, \epsilon)$-private}

\def\indvs{individuals}
\def\Indv{Individual}

\SetKwInput{KwInput}{Input}                
\SetKwInput{KwOutput}{Output}              

\newcommand{\tightsection}[1]{\vspace{-0.15cm}\section{#1}\vspace{-0.05cm}}
\newcommand{\tightsubsection}[1]{\vspace{-0.05cm}\subsection{#1}\vspace{-0.05cm}}

\ifdefined\isExtended

\newcommand{\extendedonly}[1]{#1}
\newcommand{\submissiononly}[1]{}
\else

\newcommand{\extendedonly}[1]{}
\newcommand{\submissiononly}[1]{#1}
\fi

\usepackage{multirow}

\newcommand{\fps}{\texttt{fps}}

\newcommand{\avgf}{\text{Avg}}

\newcommand{\pk}{(\rho,K)}
\newcommand{\policy}{\mathpzc{P}}

\newcommand{\rangecons}{\tilde{C_r}}
\newcommand{\sizecons}{\tilde{C_s}}

\newcommand{\Dp}[1]{\Delta_{\policy}(#1)}
\newcommand{\D}[1]{\Delta(#1)}
\newcommand{\Cr}[2]{\tilde{C_r}(#1, #2)}
\newcommand{\Cs}[1]{\tilde{C_s}(#1)}
\newcommand{\groupby}[2]{\prescript{}{#1}{\gamma}_{#2}}
\newcommand{\maxrows}{\mathsf{max\_rows}}
\newcommand{\maxchunks}[1]{\mathsf{max\_chunks}(#1)}

\newcommand{\query}{\mathbb{Q}}

\newcommand{\selectcmd}{\texttt{SELECT}}
\newcommand{\splitcmd}{\texttt{SPLIT}}
\newcommand{\processcmd}{\texttt{PROCESS}}

\usepackage[small,compact]{titlesec}
\titlespacing*{\section}{0pt}{4pt}{4pt}
\titlespacing*{\subsection}{0pt}{4pt}{4pt}
\titlespacing*{\subsubsection}{0pt}{0pt}{0pt}

\begin{document}
\newcommand{\beaver}{\includegraphics[width=12pt]{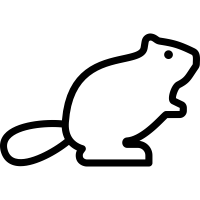}}
\newcommand{\tiger}{\includegraphics[width=12pt]{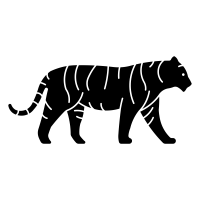}}
\newcommand{\bean}{\includegraphics[width=10pt]{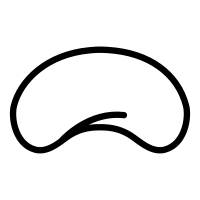}}
\newcommand{\knight}{\includegraphics[width=8pt]{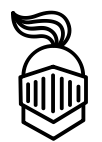}}

\title{\vspace{-20mm}\Large \bf Privid: Practical, Privacy-Preserving Video Analytics Queries\vspace{-7mm}}
\def\refMIT{\beaver}
\def\refPrinceton{\tiger}
\def\refChicago{\bean}
\def\refRutgers{\knight}
\author{
    \small
    Frank Cangialosi\refMIT{}, 
    Neil Agarwal\refPrinceton{}, 
    Venkat Arun\refMIT{}, 
    Junchen Jiang\refChicago{},
    Srinivas Narayana\refRutgers{}, 
    Anand Sarwate\refRutgers{},
    Ravi Netravali\refPrinceton{}
    \vspace{-1mm}
    \\
    \small
    \refMIT{} MIT CSAIL
    \refPrinceton{} Princeton University
    \refChicago{} University of Chicago
    \refRutgers{} Rutgers University
    \vspace{-1mm}
    \\
    \small
    privid@csail.mit.edu
\vspace{-0.1cm}
}
\date{\vspace{-8mm}}
\maketitle

\begin{abstract}

Analytics on video recorded by cameras in public areas have the potential to fuel many exciting applications, but also pose the risk of intruding on individuals' privacy. Unfortunately, existing solutions fail to practically resolve this tension between utility and privacy, relying on perfect detection of all private information in each video frame---an elusive requirement. 
This paper presents:
(1) a new notion of differential privacy (DP) for video analytics, \pkeprivacylong{}, which protects all private information visible for less than a particular duration, rather than relying on perfect detections of that information, and
(2) a practical system called \system{} that enforces duration-based privacy even with the (untrusted) analyst-provided deep neural networks that are commonplace for video analytics today. Across a variety of videos and queries, we show that \system{} achieves accuracies within 79-99\% of a non-private system.

\end{abstract}

\begin{sloppypar}
\tightsection{Introduction}

High-resolution video cameras are now pervasive in public settings~\cite{LondonCamera,infowatch-surveillance,paris-hospital, beijing-cameras,ChicagoCamera}, with deployments throughout city streets, in our doctor's offices and schools, and in the places we shop, eat, or work. Traditionally, these cameras were monitored manually, if at all, and used for security purposes, such as providing evidence for a crime or locating a missing person.
However, steady advances in computer vision~\cite{pedestrian-detection-iccv15,cnn-face-cvpr15,pyramid-network-cvpr17,facial-point-cvpr13,imagenet-classification-cacm17}
have made it possible to automate video-content analytics (both live and retrospective) at a massive scale across entire networks of cameras.
While these trends enable a variety of important applications~\cite{are-we-ready-for-ai-powered-security-cameras,powering-the-edge-with-ai-in-an-iot-world,vision-zero,smart-mall} and fuel much work in the systems community~\cite{blazeit,miris,noscope,reducto,rocket,spatula,chameleon,focus,videostorm}, they also enable privacy intrusions at an unprecedented level~\cite{epic_surveillance, stanley2019dawn}.

As a concrete example, consider the operator for a network of city-owned cameras. Different organizations (i.e., ``analysts'') want access to the camera feeds for a range of needs: (1) health officials want to measure the fraction of people wearing masks and following COVID-19 social distancing orders~\cite{analyzing-social-distancing}, (2) the transportation department wants to monitor the density and flow of vehicles, bikes, and pedestrians to determine where to add sidewalks and bike lanes~\cite{traffic-analysis}, and (3) businesses are willing to pay the city to understand shopping behaviors for better planning of promotions~\cite{retail_example}.


Unfortunately,
freely sharing the video with these parties may enable them to violate the privacy of individuals in the scene by tracking where they are, and when.
For example, the ``local business'' may actually be a bank or insurance company that wants to track individuals' private lives for their risk models, while well-known companies~\cite{google-mission-creep} or government agencies may succumb to mission creep~\cite{theyarewatching, streetlights-mission-creep}. Further, any organizations with good intentions could have employees with malicious intent who wish to spy on a friend or co-worker~\cite{aclu-video-surveillance, camera-spy-abuses}.

There is an \emph{inherent tension between utility and privacy}.
In this paper, we ask: is it possible to enable these (untrusted) organizations to use the collected video for analytics, while also guaranteeing citizens that their privacy will be protected? Currently, the answer is no. 
As a consequence, many cities have outright banned analytics on public videos, even for law enforcement purposes~\cite{camera-ban-sf,camera-ban-oak}. 

While a wide variety of solutions have been proposed (\S\ref{sec:prior}), ranging from computer vision (CV)-based obfuscation~\cite{denaturing-survey,openface,i-pic,pecam} (\eg{} blurring faces) to differential privacy (DP)-based methods~\cite{verro,videodp},
they all use some variant of the same basic strategy: 
find {\em all} private information in the video, then hide it.
Unfortunately, the first step alone can be unrealistic 
in practice (\SecNS{prior:denaturing}); it requires:
(1) an explicit specification of all private information that could be used to identify an individual (\eg{} their backpack), and then
(2) the ability to spatially \emph{locate} all of that information in \emph{every} frame of the video---a near impossible task even with state-of-the-art CV algorithms~\cite{coco_scoreboard}.
Further, 
if these approaches cannot find some private information, they fundamentally cannot \emph{know} that they missed it. Taken together, they can provide, at best, a conditional and brittle privacy guarantee such as the following:
if an individual is only identifiable by their face, and their face is detectable in every frame of the video by the implementation's specific CV model in the specific conditions of this video, then their privacy will be protected.

This paper takes a pragmatic stance and aims to provide
a definitively achievable privacy guarantee that captures the 
aspiration of prior approaches
(\ie{} individuals cannot be identified in any frame or tracked across frames)
despite the limitations that plague them.
To do this, we leverage two key observations: 
(1) a large body of video analytics queries are aggregations~\cite{blazeit,tasti}, and 
(2) they typically aggregate over durations of video (\eg{} hours or days) that far exceed the duration of any one individual in the scene (\eg{} seconds or minutes)~\cite{blazeit}.
Building on these observations, we make three contributions
by jointly designing a new notion of duration-based privacy for video analytics, a system implementation to realize it, and a series of optimizations to make it practical.

\Para{Duration-based differential privacy.}
To remove the dependence on spatially locating all private information in each video frame,
we reframe the approach to privacy 
to instead focus on the temporal aspect of private information in video data, i.e., {\em how long} something is visible to a camera. 
More specifically, building on the differential privacy (DP) framework, we propose a new notion of privacy for video,
\pkeprivacylong{} (formalized in \SecNS{definition:formal}):
\emph{anything} visible to a camera less than $K$ times for less than $\rho$ seconds each time (``\pkbounded{}'') is protected with $\epsilon$-DP~\cite{original-dp}.
Regardless of the video owner's underlying privacy policy, they express it only through an appropriate $\pk$ that captures information they deem private.
For example, if they choose $\pk$ such that \emph{all} individuals are visible for less time, then \pkeprivacy{} prevents an analyst from determining whether or not \emph{any} single individual appeared at any time, which in turn precludes tracking them. We discuss other policies in \Sec{definition:usage}.

This notion of privacy has three benefits. 
First, it decouples the definition of privacy from its enforcement. The enforcement mechanism does not need to make any decisions about what is private or find private information to protect it; everything (private or not) captured by the bound is protected. 
Second, a $\pk$ bound that captures a set of individuals implicitly captures and thus protects any information visible for the same (or less) time without specifying it (\eg{} an individual's backpack, or even their gait).
Third, protecting all individuals in a video scene requires only their maximum duration, and estimating this value is far more robust to the imperfections of CV algorithms than precisely locating those individuals and their associated objects in each frame. 
For example, even if a CV algorithm misses individuals in some frames (or entirely), it can still capture a representative sample and piece together trajectories well enough to estimate their duration (\SecNS{definition:usage}).

\Para{Privid: a differentially-private video analytics system.}
Realizing \pkeprivacy{} (or more generally, any DP mechanism) in today's video analytics pipelines faces several challenges. 
In traditional database settings, implementing DP requires adding random noise proportional to the \emph{sensitivity} of a query, \ie{} the maximum amount that any one piece of private information could impact the query output.
However, bounding the sensitivity
is difficult in video analytics pipelines because (1) pipelines typically operate as bring-your-own-query-implementation to support the wide-ranging applications described earlier~\cite{cv-frameworks,rocket,amazon-rekognition,google-cloud-vision,msft-cv,azure-face,ibm-maximo}, and (2) these implementations involve video processing algorithms that increasingly rely on deep neural networks (DNNs), which are notoriously hard to introspect or vet (and thus, trust). 


To bound the sensitivity necessary for \pkeprivacy{} while supporting untrusted analyst-provided DNNs,
\system{} accepts analyst queries structured in a {\em split-process-aggregate} format.
In particular, (i) videos are split temporally into contiguous chunks, (ii) each chunk of video is processed by an arbitrary analyst-provided DNN to produce an (untrusted) table, (iii) an aggregation over the table computes a result, and (iv) noise is added to the result before release. 
In \Sec{system:sensitivity}, we show how this structure enables \System to relate the duration (number of chunks) where an entity is visible to the appropriate amount of noise to add to the aggregate result.
This in turn allows \system{} to bound noise without relying on a trusted table, something prior DP query mechanisms have not had to do~\cite{privatesql,flex,google-dp}.

\Para{Optimizations for improved utility.} 
To further enhance utility, \system{} provides two video-specific optimizations to lower the required noise while preserving an equivalent level of privacy: (i) the ability to mask regions of the video frame, (ii) the ability to split frames spatially into different regions, and aggregate results from these regions. These optimizations result in limiting the portion of the aggregate result that any individual’s presence can impact, enabling a ``tighter'' $\pk$ bound and in turn a higher quality query result.

\Para{Evaluation.}
We evaluate \system{} using a variety of public videos and a diverse range of queries inspired by recent work in this space.
\system{} achieves accuracy within 79-99\% of a non-private system, while satisfying an instantiation of \pkeprivacy{} that protects all individuals in the video.
We discuss ethics in \Sec{ethics}, and will open-source \system{} upon publication.




\section{Problem Statement}
\label{sec:problem-statement}

\tightsubsection{Video Analytics Background} 
\label{sec:video-analytics}

Video analytics pipelines are employed to answer high-level questions about segments of video captured
from one or more cameras and across a variety of time ranges. Example questions include ``how many people entered store X each hour?'' or ``which roads housed the most accidents in 2020?'' A question is expressed as a \emph{query}, which encompasses all of the computation necessary to answer that question.\footnote{Our definition is distinct from related work, which defines a query as returning intermediate results (\eg{} bounding boxes) rather than the final answer to the high-level question.} For example, to answer the question ``what is the average speed of red cars traveling along road Y?'', the query would include an object detection algorithm to recognize cars, an object tracking algorithm to group them into trajectories, an algorithm for computing speed from a trajectory, and logic to filter only the red cars and average their speeds.

%

\tightsubsection{Problem Definition}
\label{sec:problem}

Video analytics pipelines broadly involve four logical roles (though any combination may pertain to the same entity):
\begin{itemize}
    \item \textbf{\Indv{}s}, whose behavior and activity are observed by the camera.
    \item \textbf{Video Owner}, who operates the camera and thus the video data it captures. 
    \item \textbf{Analyst}, who wishes to run queries over the video. 
    \item \textbf{Compute Provider},
    who executes the analyst's query.
\end{itemize}


\noindent In this work, we are concerned with the dilemma of a video owner. 
The owner would like to enable a variety of (untrusted) analysts to answer questions about its videos (such as those in~\Sec{video-analytics}), as long as the results do not infringe on the privacy of the individuals who appear in the videos.
Informally, privacy ``leakage'' occurs when an analyst can learn something about a specific individual that they did not know before executing a query. To practically achieve these properties, a system must meet three concrete goals:
\begin{enumerate}
    \item \textbf{Formal notion of privacy}. The system's privacy policies should formally describe the type and amount of privacy that could be lost through a query. Given a privacy policy, the system should be able to provide a \emph{guarantee} that it will be enforced, regardless of properties of the data or query implementation.

    \item \textbf{Maximize utility for analysts.} The system should support queries whose final \emph{result} does not infringe on the privacy of any individuals. Further, if accuracy loss is introduced to achieve privacy for a given query, it should be possible to bound that loss (relative to running the same query over the original video, without any privacy preserving mechanisms). Without such a bound, analysts would be unable to rely on any provided results.

    
    
    \item \textbf{``Bring Your Own Model''}. Computer vision models are at the heart of modern video processing. However, there is not one or even a discrete set of models for all tasks and videos. Even the same task may require different models, parameters, or post-processing steps when applied to different videos. In many cases, analysts will want to use models that they trained themselves, especially when training involves proprietary data. Thus, a system must allow analysts to provide their own video-processing models. 
\end{enumerate}

It is important to note that the analytics setting we seek to enable is distinct from \emph{video security} (\eg{} finding a stolen car or missing child), which \emph{requires} identification of a particular individual, and is thus directly at odds with individual privacy. 
In contrast, analytics queries involve searching for patterns and trends in large amounts of data; intermediate steps may operate over the data of specific individuals, but they do not distinguish individuals in their final aggregated result (\SecNS{video-analytics}).

\tightsubsection{Threat Model}
\label{sec:threat-model}
The video owner employs a privacy-preserving system to handle queries about a set of cameras it manages; the system retains full control over the video data, analysts can only interact with it via the query interface.
The video owner does not trust the analysts. Any number of analysts may be malicious and may collude to violate the privacy of the same individual.
However, analysts trust the video owner to be honest.
Analysts are also willing to share their query implementation (so that the video owner can execute it). 

Analysts pose queries adaptively (\ie{} the full set of queries is not known ahead of time, and analysts may utilize the results of prior queries when posing a new one). 
A single query may operate over video from multiple cameras, but only those under the control of the same owner. We the video owner has sufficient computing resources to execute the query, either via resources that they own, or through the secure use of third-party resources~\cite{visor}.

In summary, we focus entirely on managing the leakage of information about individuals from the video owner to the analyst, and consider any other privacy or trust concerns between the parties to be orthogonal. 


\tightsection{Limitations of Related Work}
\label{sec:prior}
Before presenting our solution, we consider prior privacy-preserving mechanisms (both for video and in general). Unfortunately, each fails to satisfy at least one of the goals in~\S\ref{sec:problem}.

\tightsubsection{Denaturing}
\label{sec:prior:denaturing}

The predominant approach to privacy preservation with video data is \emph{denaturing}~\cite{denaturing-survey,openface,i-pic,pecam,prisurv,privacycam}, whereby systems aim to obscure (\eg{} via blurring~\cite{i-pic} or blocking~\cite{openface} as in~\Fig{denature:fails}) any private information in the video before releasing it for analysis. In principle, if nothing private is left in the video, then privacy concerns are eliminated.

\par
The fundamental issue is that denaturing approaches require \emph{perfectly} accurate and comprehensive knowledge of the spatial locations of private information in \emph{every frame} of a video. Any private object that goes undetected, even in just a single frame, will not be obscured and thus directly leads to a leakage of private information.

To detect private information, one must first semantically define \emph{what} is private, \ie{} what is the full set of information linked, directly or indirectly, to the privacy of each individual? While some information is obviously linked (\eg{} an individual's face), it is difficult to determine \emph{all} such information for all individuals in all scenarios. For instance, a malicious analyst may have prior information that a video owner does not, such as knowledge that a particular individual carries a specific bag or rides a unique bike (\eg{}~\Fig{denature:fails}-B). Further, even with a semantic definition, detecting private information is difficult. State-of-the-art computer vision algorithms commonly miss objects or produce erroneous classification labels in favorable video conditions~\cite{zhu2017flow}; performance steeply degrades in more challenging conditions such as poor lighting, distant objects, and low resolution, all of which are common in public video. Taken together, the problem is that denaturing systems cannot guarantee whether or not a private object was left in the video, and thus fail to provide a formal notion of privacy (violating Goal 1).

Denaturing also falls short from the analyst's perspective. 
First, it inherently precludes (safe) queries that aggregate over private information (violating Goal 2). For example, an urban planner may wish to count the number of people that walk in front of camera A and then camera B. Doing so requires identifying and cross-referencing individuals between the cameras (which is not possible if they have been denatured), but the aggregate count may be large and safe to release.\footnote{As a workaround, the video owner could annotate denatured objects with query-specific information, but this would conflict with Goal 3.} Second, obfuscated objects are not naturally occurring and thus video processing pipelines are not designed to handle them. If the analyst's processing code and models have not been trained explicitly on the type of obfuscation the video owner is employing, it may behave in unpredictable and unbounded ways (violating Goal 2). 

\begin{figure}
  \centering
  \includegraphics[width=0.87\columnwidth]{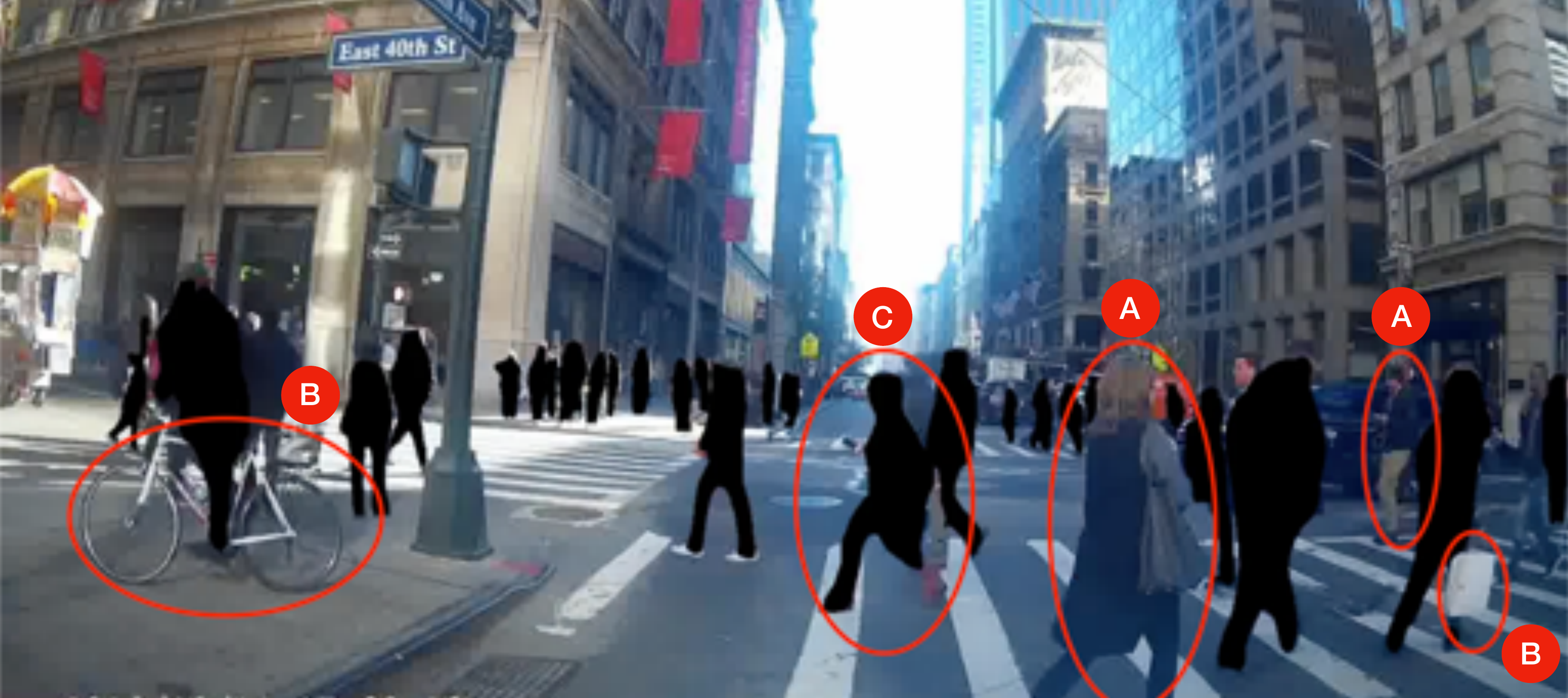}
  \vspace{5pt}
  \tightcaption{
      A video clip after (silhouette) denaturing exemplifying some of its shortcomings: (A) entirely missed detections, (B) potentially-identifying objects not incorporated in privacy definition, (C) silhouette may reveal gait.
    }
  \label{fig:denature:fails}
\end{figure}

\tightsubsection{Differential Privacy}
\label{sec:prior:dp}

Differential Privacy (DP) is a strong formal definition of privacy from the databases community~\cite{original-dp}. 
It enables analysts to compute aggregate statistics over a database, while protecting the presence of any individual entry in the database. 
DP is not a privacy-preserving mechanism itself, but rather a goal that an algorithm can aim to satisfy.
Informally speaking, an algorithm satisfies DP if adding or removing an individual from the input database does not noticeably change the output of computation, almost as if any given individual were not present in the first place. More precisely,


\begin{defn}
Two databases $D$ and $D'$ are \emph{neighboring} if they differ in the data of only a single user (typically, a single row in a table).
\end{defn}
\begin{defn}
A randomized algorithm $\mathcal{A}$ is \emph{$\epsilon$-differentially private} if, for all pairs of neighboring databases $(D, D')$ and all $S \subseteq Range(\mathcal{A})$: 
\begin{equation}
\vspace{-5pt}
\mathrm{Pr}[\mathcal{A}(D) \in S] \le e^\epsilon \mathrm{Pr}[\mathcal{A}(D') \in S]
\label{eq:classic-dp}
\end{equation}
\end{defn}


A non-private computation (\eg{} computing the sum of all bank balances) is typically made differentially private by adding random noise sampled from a Laplace distribution to the final result of the computation~\cite{original-dp}. The scale of noise is set proportional to the \emph{sensitivity} ($\Delta$) of the computation, or the maximum amount by which the computation's output could possibly change due to the presence/absence of any one individual. 
For instance, suppose a database contains a value $v_i \in V$ for each user $i$, where $l \le v_i \le u$. If a query seeks to sum all values in $V$, any one individual's value can influence that sum by at most $\Delta = u - l$, and thus adding noise with scale $u-l$ would satisfy DP. \extendedonly{From an analyst's perspective, given a pair of outputs generated from neighboring databases, they cannot determine how much, if any, of the difference is due to the true data as opposed to the noise.}

\Para{Challenges.} Determining the sensitivity of a computation is the key ingredient of satisfying DP. It requires understanding (a) how individuals are delineated in the data, and (b) how the aggregation incorporates information about each individual. In the tabular data structures that DP was designed for, these are straightforward. Each row (or a set of rows sharing a unique key) typically represents one individual, and queries are expressed in relational algebra, which describes exactly how it aggregates over these rows. However, these answers do not translate to video data; we next discuss the challenges in the context of several applications of DP to video analytics.

Regarding \emph{requirement (a)}, as described in~\Sec{prior:denaturing}, it is difficult and error-prone to determine the full set of pixels in a video that correspond to each user (including all potentially identifying objects). Accordingly, prior attempts of applying DP concepts to video analytics~\cite{verro, videodp} that rely on perfectly defined and detected private information (via computer vision) fall short in the same way as denaturing approaches (violating Goal 1). 

Regarding \emph{requirement (b)},
arbitrary video processing executables (such as deep neural networks) are not transparent about how they incorporate information about private objects into their results. Thus, without a specific query interface, the ``tightest'' possible bound on the sensitivity of an arbitrary computation over a video is simply the entire range of the output space. In this case, satisfying DP would add noise greater than or equal to any possible output, precluding any utility (violating Goal 2).

Given that DP is well understood for tables, a natural idea would be for the video owner to use their own (trusted) model to first convert the video into a table (\eg{} of objects in the video), then provide a DP interface over \emph{that table}\footnote{Note: this strawman is analogous to the video owner adding DP on top of an existing video analytics system such as~\cite{blazeit, miris}, and the same arguments for lack of flexibility apply.} (instead of directly over the video itself).
However, in order to provide a guarantee of privacy, the video owner would need to completely trust the model that creates the table. This would be difficult even for a model created by the video owner themselves (any inaccurate outputs could undermine the guarantee), but more importantly it entirely precludes using a model created by the untrusted analyst (violating Goal 3).

\section{Our Approach}
\label{sec:approach}



The shortcomings of the previous approaches arise from the concern that neither the appearances of individuals in a video nor their impact on query output is known or easy to identify automatically and precisely. 
To address these shortcomings, we co-design a new notion of duration-based privacy and a new video analytics system, \system{}, to implement it. 
The key idea of our solution is to shift from delineating information spatially (where each object is in each frame) to delineating information {\em temporally} (how long an individual is visible for).
Instead of defining $\epsilon$-DP based on individuals directly, duration-based privacy provides $\epsilon$-DP for {\em anything visible for less than a given duration}, which can in turn be used as a proxy for protecting individuals.
This section provides an intuitive overview of these components and why they together meet the two requirements of DP-based privacy (\SecNS{prior:dp}). 


\Para{New notion of duration-based privacy:}
With this definition, the contribution of each object to the system's output is bounded by the duration the object is visible in the video,
in terms of the number of times it appears ($K$) and the duration of each appearance ($\rho$).
We will formally define these terms in the next section.

\smallskip

\noindent{\em Implication: Definition of privacy decoupled from its enforcement.}
\pkeprivacy{} directly addresses requirement (a): 
recall from \SecNS{prior:denaturing} that the traditional privacy definition is difficult to enforce in video analytics because locating (and obscuring) {\em all} appearances of each individual is impractical.
In contrast, 
\pkeprivacy{} provides a simple {\em interface} for the video owner to express their privacy goal via the parameters $\pk$, which delineate how each individual contributes to the video.
Consequently, we free the enforcement mechanism from making any decisions about what is private (\eg detecting/blocking all appearances of each individual).

\Para{An execution framework to enforce it:}
\system{} adopts a split-process-aggregate framework to enforce \pkeprivacy{}: the video is split (temporally) into individual chunks, each chunk is independenly processed using the analyst-provided executable, and finally the results are aggregated using a familiar SQL-like statement.

\smallskip
\noindent{\em Implication: Agnostic to query model. } 
\system's split-process-aggregate framework directly addresses requirement (b): by forcing queries to be explicit about this format, it ties the amount an event can impact the output to its duration, rather than the particular output of the model on any individual frame.
Since video analytics pipelines typically process one frame (or a small group of frames) at a time (as opposed to the entire video in one shot) and use separate logic to aggregate per-frame results, \system's split-process-aggregate framework is naturally compatible to many (albeit not all) queries.

\Para{Amenable to off-the-shelf computer vision:}
Choosing meaningful values of $\pk$ that capture the video owner's privacy goals require some understanding of the particular scene captured by a camera. For example, protecting the privacy of all individuals requires knowing the maximum amount of time an individual could be visible to that camera. This could come from domain knowledge, but as we will see in \SecNS{definition:usage}, estimating a single maximum duration is feasible even with imperfect off-the-shelf CV algorithms and is considerably easier than specifying all private information of each individual in each frame. Intuitively, this makes sense: even if an individual is missed in a few frames, a tracking algorithm could still link their trajectory and correctly estimate their duration, and even if some individuals are missed entirely, a representative sample over a sufficiently long time period should closely capture the maximum.

\if
{

This ensures that the amount an object can impact the output is bounded solely by its duration. Since it does not depend on the output of the model for any individual frame, the model itself does not need to be trusted (and thus can be provided by the untrusted analyst).

{\em 3. Amenable to current video analytics pipelines.} As described in \Sec{video-analytics}, videos are typically analyzed by using CV to process one frame at a time (rather than, say, a single pixel or row of pixels), and separate logic is used to aggregate these per-frame results. This enables a design for \system{} which addresses requirement (b): by forcing queries to be explicit about this split-process-aggregate format, it ensures that the amount an object can impact the output is bounded solely by its duration. Since it does not depend on the output of the model for any individual frame, the model itself does not need to be trusted (and thus can be provided by the untrusted analyst).

In short, how does it achieve the three goals?

\pkeprivacy{} directly addresses requirement (a): informally, it says an object's contribution to the output is directly proportional to the amount of time it is visible in the video, and provides $\epsilon$-DP for anything visible for less than a given duration.
\footnote{To clarify, \pkeprivacy{} is \emph{not} a departure from DP, but rather a concrete decision about what to protect in the context of video, where the foundation of DP, an ``individual'', is not well-defined.}
We choose this definition and delinate information in a video temporally (rather than spatially) for three reasons.

\emph{1. Unconditional guarantee.} 
In contrast to the traditional notion of DP which aims to protect all ``individuals'' and internalizes the definition of how much they can contribute to the database, \pkeprivacy{} asks the video owner to define how individuals contribute to the video via the duration parameters $\pk$.
This frees the mechanism from making any decisions about what is private, and thus the guarantee is not conditional on the specific system implementation or target video data.

\emph{2. Pracitcal given today's CV algorithms.} 
To capture the traditional notion of DP (indistinguishability for any individual), the parameters $\pk$ must be set to the maximum duration an individual could appear in the video.
In practice, today's CV algorithms are capable\footnote{We choose conservative parameters to aid this, detailed in~\App{cv-params}.}
of producing a good estimate of this maximum duration, despite their imperfections which impede prior approaches.
Intuitively, this makes sense: a single upper-bound is coarser and amenable to seeing only a sample of individuals, compared to prior approaches which depend on finding every instance of all information tied to each individual.

We provide some validation of this intuition over three representative videos from our evaluation. 
For each video, we choose a 10-minute segment and manually annotate the duration of each individual\footnote{Any object with class \texttt{person}, \emph{or} a vehicle class, explained in \Sec{definition:usage}.}
(``Ground Truth''),
then use state-of-the-art object detection and tracking to estimate the durations and report the maximum (``CV'').
Our results, summarized in Table~\ref{tab:cv-effectiveness}, show that, while object detection misses a non-trivial fraction of bounding boxes in each video, when combined with tracking, it conservatively estimates the maximum duration. In other words, for each video, using these algorithms to parameterize a \pkeprivate{} system would successfully protect the privacy of \emph{all} individuals, while using them to implement a prior approach would not. 

\begin{figure}
  \centering
  \includegraphics[width=\columnwidth]{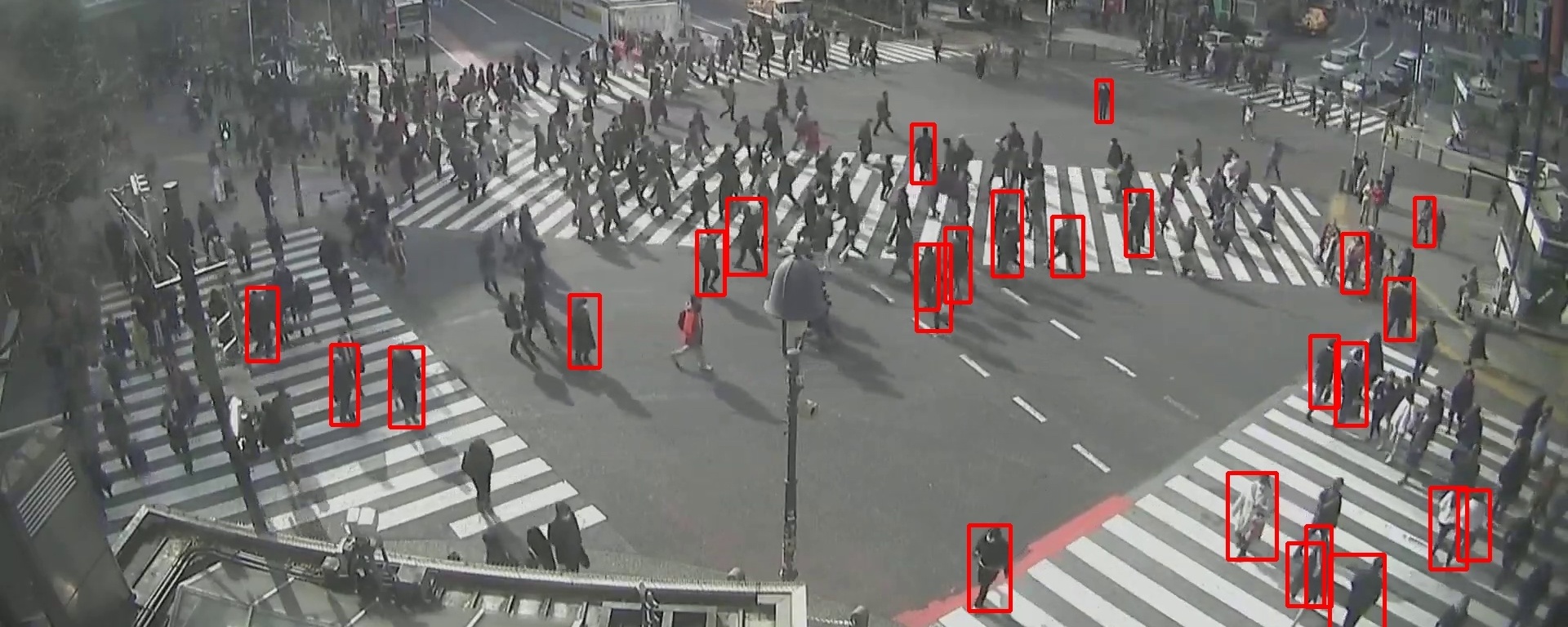}
  \vspace{5pt}
  \tightcaption{
  The results of a state-of-the-art object detection algorithm (filtered to ``person'' class) on one frame of \shibuya{}. The algorithm misses 76\% of individuals in the frame, but is \emph{still} able to produce a conservative bound on the maximum duration of all individuals (Table~\ref{tab:cv-effectiveness}).}
  \label{fig:shibuya-missed}
\end{figure}

\begin{table}[]
\centering
\small
\begin{tabular}{|l|l|l|l|}
\hline

\multirow{2}{*}{\textbf{Video}}    & 
\multicolumn{2}{c|}{\textbf{Maximum Duration}} &
\multirow{2}{*}{\textbf{\shortstack{\% Objects \\ CV Missed}}} \\ 
    \cline{2-3}
    & Ground Truth & CV Estimate & \\ 
    \hline
\auburn{}   & 81 sec              & 83 sec  & 29\%   \\ \hline
\hampton{}*  & 316 sec               & 439 sec    & 5\%   \\ \hline
\shibuya{}  & 270 sec               & 354 sec    & 76\%   \\ \hline
\end{tabular}
\vspace{15pt}
\tightcaption{Despite the imperfection of current CV algorithms (exemplified by \% objects they failed to detect), they still produce a conservative estimate on the duration of any individual's presence. *For the purposes of this experiment, we ignored cars that were parked for the entire duration of the segment.\vspace{-5pt}}
\label{tab:cv-effectiveness}
\end{table}

\emph{3. Amenable to current video analytics pipelines.} As described in \Sec{video-analytics}, videos are typically analyzed by using CV to process one frame at a time (rather than, say, a single pixel or row of pixels), and separate logic is used to aggregate these per-frame results. This enables a design for \system{} which addresses requirement (b): by forcing queries to be explicit about this split-process-aggregate format, it ensures that the amount an object can impact the output is bounded solely by its duration. Since it does not depend on the output of the model for any individual frame, the model itself does not need to be trusted (and thus can be provided by the untrusted analyst).

}
\fi


\section{Event Duration Privacy}
\label{sec:definition}

We will first formalize \pkeprivacy{}, then provide the intuition for what it protects and clarify its limitations.

\tightsubsection{Definition}
\label{sec:definition:formal}

We consider a video $V$ to be an arbitrarily long sequence of frames, sampled at $f$ frames per second, recorded directly from a camera (\ie{} unedited). A ``segment'' $v \subset V$ of video is a contiguous subsequence of those frames. The ``duration'' of a segment $d(v)$ is measured in real time (seconds), as opposed to frames.\extendedonly{, \ie{} $d(v) = \frac{\abs{v}+1}{f}$.}
An ``event'' $e$ is abstractly \emph{anything} that is visible within the camera's field of view.

As a running example, consider a video segment $v$ in which individual $x$ is visible for 30 seconds before they enter a building, and then another 10 seconds when they leave some time later. The ``event'' of $x$'s visit is comprised of one 30-second segment, and another 10-second segment.

\begin{defn}[\pkbounded{} events]
\label{def:pkbounded}
An event $e$ is \pkbounded{} if there exists a set of $\le K$ video segments that completely contain\footnote{If a set of segments completely contain an event, then that event is not visible in any frames outside of those segments. We note that, for this definition, visibility does not depend on whether or not a computer vision algorithm or even a human is capable of recognizing it.} the event, and each of these segments individually have duration $\le \rho$.
\end{defn}

\noindent (Ex). The tightest bound on $x$'s visit is
$(\rho=30s,K=2)$. To be explicit, $x$'s visit is also \pkbounded{} for any $\rho\ge30s$ and $K\ge2$. 

\begin{defn}[\pkneighboring{} videos]
\label{def:pkneighboring}
Two video segments $v, v'$ are \pkneighboring{} if 
the set of frames in which they differ is \pkbounded{}.
\end{defn}
\noindent (Ex). One potential $v'$ is a hypothetical video in which $x$ was never present (but everything else observed in $v$ remained the same). Note this is purely to denote the strength of the guarantee in the following definition, the video owner does not actually construct such a $v'$. 

\begin{defn}[\pkeprivacylong{}]
\label{def:pkeprivacy}
A randomized mechanism $\mathcal{M}$ satisfies \pkeprivacylong{}
\footnote{We chose to use $\epsilon$-DP rather than the more general $(\epsilon,\delta)$-DP for simplicity, since the difference is not significant to our notion. All of our concepts and results could be trivially extended to $(\epsilon,\delta)$-DP without any additional insights.}
iff for all possible pairs of \pkneighboring{} videos $v,v'$, any finite set of queries $Q = \{q_1, q_2, ...\}$ and all $S_q \subseteq Range(\mathcal{M}(\cdot, q))$:
\begin{align*}
Pr[(\mathcal{M}(v, q_1),\ldots,\mathcal{M}(v, q_n)) \in S_{q_1} \times \cdots \times S_{q_n}] \le \\ e^{\epsilon}Pr[(\mathcal{M}(v', q_1),\ldots,\mathcal{M}(v', q_n))) \in S_{q_1} \times \cdots \times S_{q_n}]
\end{align*}
\end{defn}



\Para{Guarantee.}
\pkeprivacy{} protects all \pkbounded{} events (such as $x$'s visit to the building) with $\epsilon$-DP: informally, if an event is \pkbounded{}, an adversary cannot increase their knowledge of whether or not the event happened by observing a query result from $\mathcal{M}$.
To be clear, \pkeprivacy{} is \emph{not} a departure from DP, but rather an extension to explicitly specify what to protect in the context of video.	

%
%

\tightsubsection{Choosing a Privacy Policy}\label{sec:definition:usage}
The video owner is responsible for choosing the parameter values $\pk$ (``policy'') that bound the class of events they wish to protect. 
They may use domain knowledge, employ CV algorithms to analyze durations in past video from the camera, or a mix of both.
Regardless, they express their goal to \system{} solely through their choice of $\pk$.

In this paper, we focus on the privacy goal of standard DP:
protect the \emph{appearance} of all individuals.~\footnote{In order to fully capture this intent, we also seek to protect the appearance of all vehicles, since vehicles often uniquely identify their driver (sometimes even if the license plate is not visible).}
In other words, a query should not be able to determine whether or not any individual appeared in the video, directly or indirectly.

\begin{figure}
  \centering
  \includegraphics[width=\columnwidth]{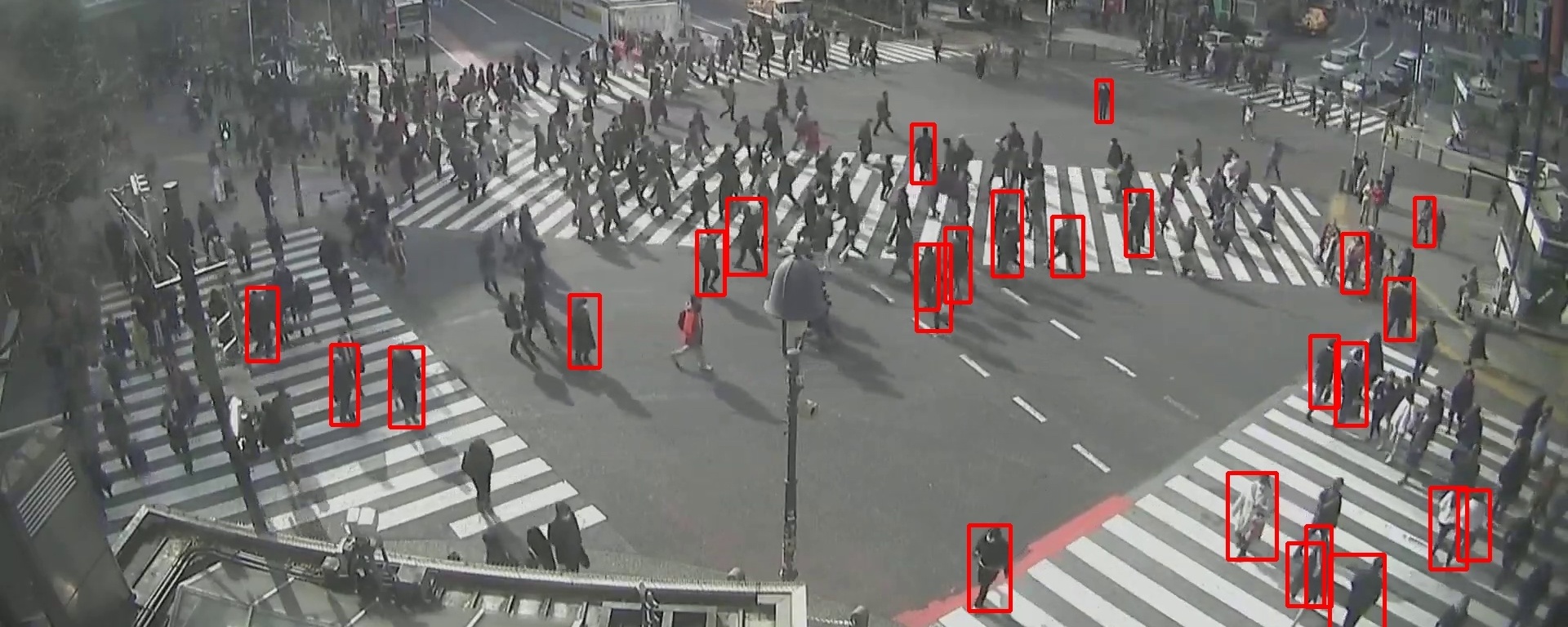}
  \vspace{5pt}
  \tightcaption{
  The results of a state-of-the-art object detection algorithm (filtered to ``person'' class) on one frame of \shibuya{}. The algorithm misses 76\% of individuals in the frame, but is \emph{still} able to produce a conservative bound on the maximum duration of all individuals (Table~\ref{tab:cv-effectiveness}).}
  \label{fig:shibuya-missed}
\end{figure}

\begin{table}[]
\centering
\small
\begin{tabular}{|l|l|l|l|}
\hline
\multirow{2}{*}{\textbf{Video}}    & 
\multicolumn{2}{c|}{\textbf{Maximum Duration}} &
\multirow{2}{*}{\textbf{\shortstack{\% Objects \\ CV Missed}}} \\ 
    \cline{2-3}
    & Ground Truth & CV Estimate & \\ 
    \hline
\auburn{}   & 81 sec              & 83 sec  & 29\%   \\ \hline
\hampton{}*  & 316 sec               & 439 sec    & 5\%   \\ \hline
\shibuya{}  & 270 sec               & 354 sec    & 76\%   \\ \hline
\end{tabular}
\vspace{15pt}
\tightcaption{Despite the imperfection of current CV algorithms (exemplified by \% objects they failed to detect), they still produce a conservative estimate on the duration of any individual's presence. *For the purposes of this experiment, we ignored cars that were parked for the entire duration of the segment.\vspace{-5pt}}
\label{tab:cv-effectiveness}
\end{table}

\Para{Automatic setting of $\pk$.}
The primary reason \pkeprivacy{} is \emph{practical} is that, despite their imperfections, today's CV algorithms are capable of producing good estimates of the maximum duration any individuals are visible in a scene. 
We provide some evidence of this intuition over three representative videos from our evaluation. 
For each video, we chose a 10-minute segment and manually annotate the duration of each individual (\texttt{person} or \texttt{vehicle}),
\ie ``Ground Truth'',
then use state-of-the-art object detection and tracking to estimate the durations and report the maximum (``CV'').
Our results, summarized in Table~\ref{tab:cv-effectiveness}, show that, while object detection misses a non-trivial fraction of bounding boxes, when combined with tracking, it conservatively estimates the maximum duration. In other words, for our three videos, using these algorithms to parameterize a \pkeprivate{} system would successfully protect the privacy of \emph{all} individuals, while using them to implement any prior approach would not.


\Para{Alternative policies.}

\noindent \emph{Relaxing the set of private individuals:}
Sometimes protecting \emph{all} individuals may be unnecessary. 
For instance, consider a camera in a store;
employees will appear significantly longer and more frequently than customers (\eg{} 8 hours every day vs. 30 minutes once a week), but if the fact that the employees work there is already public knowledge, the video owner could relax their goal to only bound only the appearance of customers (with smaller $\rho$ and $K$).

\noindent \emph{Generic policies:} Rather than protecting the presence of individuals, a policy could be used to simply limit analysts from collecting detailed information. Consider a policy $(\rho=5m, K=1)$. Suppose individual $x$ stops and talks to a few people on their way to work each morning, but each conversation lasts less than 5 minutes. Although the policy does not protect $x$'s presence or even the fact that they often stop to chat on their way to work, it does protect each conversations; an analyst cannot determine who $x$ spoke to or about what.

\tightsubsection{Flexible privacy guarantees} \label{sec:definition:tradeoff}

As is standard with DP, anything not \pkbounded{} is not immediately revealed in the clear, but rather experiences a ``graceful degradation'' of privacy as a function of how far it is from $\rho$ and $K$.
As an event is visible for more time than $\rho$ or appears for more than $K$ segments, the probability an analyst can detect whether or not it happened increases (put differently, the effective $\epsilon$ increases, providing a weaker guarantee of privacy). Thus, an event that only exceeds the $\pk$ bound marginally is still afforded almost an equivalent level of privacy.
This in effect {\em relaxes} the requirement that $\pk$ be set strictly to the maximum duration an individual could appear in the video to achieve useful levels of privacy.
We provide a formalization of this degradation and visualize it as a function of appearance duration in~\App{graceful-degradation}. 

An alternative interpretation of the definition is that 
$\epsilon$ defines the level of privacy guaranteed \emph{for events of certain duration}. 
For instance, \pkeprivacy{} protects a $(\rho,2K)$-bounded event with $2\epsilon$-DP (weaker), or a $(\rho,\frac{1}{2}K)$-bounded event with $\frac{1}{2}\epsilon$-DP (stronger).
The relationship with $\rho$ is roughly the same, but the constants depend upon the mechanism details.\extendedonly{; we derive it for \system{} in~\App{pke-proportional}. }

Thus,
\pkeprivacy{} extends (albeit in a weaker form) to queries that span long time windows. 
The larger the time window of video a query analyzes, the more times an individual may appear (even if each appearance is itself bounded by $\rho$). 
Consider our example individual $x$ and policy $(\rho=30s,K=2)$ from~\Sec{definition:formal}.
In the query window of a single day $d$, they appear twice; they are properly \pkbounded{} and thus the event ``individual $x$ appeared on day $d$'' is protected with $\epsilon$-DP. 
Now, consider a query window of one week; $x$ appears 14 times (2 times per day), so the event ``individual $x$ appeared sometime this week'' is $(\rho, 7K)$-bounded and thus protected with (weaker) $7\epsilon$-DP.  However, the more specific event ``individual $x$ appeared on day $d$ (for any $d$ in the week)'' is \emph{still} $(\rho,K)$-bounded, and thus still protected with the same $\epsilon$-DP. 
In other words, as queries increase their window size, they can detect an individual's presence \emph{somewhere} in the window with greater certainty, but at the cost of coarser temporal granularity. 

\section{\System{}}
\label{sec:system}

In this section, we present \system{}, a privacy-preserving video analytics system that satisfies \pkeprivacy{} (\SecNS{problem} Goals 1 and 2) and provides an expressive query interface which allows analysts to supply their own (untrusted by \system{}) video-processing code (Goal 3). 

\tightsubsection{Overview} 
\label{sec:system:overview}

\system{} supports \emph{aggregation} queries, which process a ``large'' amount of video data (\eg{} several hours/days of video) and produce a ``small'' number of bits of output (\eg a few 32-bit integers). Examples of such tasks include counting the total number of individuals that passed by a camera in one day, or computing the average speed of cars observed. In contrast, \system{} does not support a query such as reporting the location (\eg{} bounding box) of an individual or car within the video frame. \System{} can be used for one-off ad-hoc queries or standing queries running over a long period, \eg{} the total number of cars per day, each day over a year.

The video owner decides the level of privacy provided by \system. The video owner chooses a privacy policy 
$\pk$
and privacy budget ($\epsilon$) for each camera they manage. 
Given these parameters, \system{} provides a guarantee of \pkeprivacy{} (Theorem~\ref{thm:multiple-queries}) for all queries over all cameras it manages.

To satisfy the privacy guarantee, \system{} utilizes the standard Laplace mechanism from DP~\cite{original-dp} to add random noise to the aggregate query result before returning the result to the analyst. The key technical pieces of \system{} are: (i) providing analysts the ability to specify queries using arbitrary untrusted code (\SecNS{system:query}), (ii) adding noise to results to guarantee \pkeprivacy for a single query (\SecNS{system:sensitivity}), and (iii) extending the guarantee to handle multiple queries over the same cameras (\SecNS{budget}).

\tightsubsection{\system{} Query Structure}
\label{sec:system:query}

\Para{Execution model.} \system{} structures queries using a split-process-aggregate approach in order to tie the duration of an event to the amount it can impact the query output. The target video is split temporally into chunks, then each chunk is fed to a separate instance of the analyst's processing code, which outputs a set of rows. Together, these rows form a traditional tabular database (untrusted by \system{} since it is generated by the analyst). The aggregation stage runs a SQL query over this table to produce a raw result. Finally, \system{} adds noise (\SecNS{system:sensitivity}) and returns \emph{only} the noisy result to the analyst, not the raw result or the intermediate table.

\Para{Query interface.}
An analyst must submit to \system{} both a written query and video processing executable(s). A query contains one or more of \emph{each} of the 3 following statements:

\noindent $\bullet$ \splitcmd{} statements choose a segment of video (camera, start and end datetime) as input, and produce a set of video chunks as output. They specify how the segment should be split into chunks, \ie{} the chunk duration and stride between chunks. 
\extendedonly{In \Sec{practical:mask}, we present an optional optimization that allows analysts to split the video spatially in addition to (not instead of) temporally. This results in multiple chunks per time unit, one for each spatial region. For example, splitting a 100 minute video into 1 minute chunks with 4 spatial regions would result in a set of 400 chunks.}

\noindent $\bullet$ \processcmd{} statements take a set of \splitcmd{} chunks as input, and produce a traditional (``intermediate'') table as output. They specify which executable should be used to process the chunks, the schema of the resulting table, and the maximum number of rows output by each chunk ($\maxrows$, necessary to bound the sensitivity, \SecNS{system:sensitivity}).

\noindent $\bullet$ \selectcmd{} statements resemble typical SQL \selectcmd{} statements that operate over \processcmd{} tables and output a \pkeprivate{} result. 
They must have an aggregation as the final operation (though it may be part of a \texttt{GROUPBY}).
\system{} supports the standard aggregation functions (\eg{} \texttt{COUNT}, \texttt{SUM}, \texttt{AVG}) and the core set of typical operators  as internal relations.
An aggregation must specify the range of each column it aggregates (just as in related work on DP for SQL~\cite{privatesql}).
Each \selectcmd{} constitutes at least one data release: one for a single aggregation or multiple for a \texttt{GROUPBY} (one for each key). Each data release receives its own sample of noise and consumes additional privacy budget (\SecNS{budget})

In order to aggregate across multiple video sources (separate time windows and/or multiple cameras), the query can use a \splitcmd{} and \processcmd{} for each video source, and then aggregate using a \texttt{JOIN} and \texttt{GROUPBY} in the \selectcmd{}.

\begin{figure}
\begin{lstlisting}[language=SQL, breaklines=true, basicstyle=\ttfamily\scriptsize, morekeywords={SPLIT, BEGIN, END, CHUNK, BY TIME, STRIDE, INTO, PROCESS, USING, TIMEOUT, PRODUCING, SCHEMA, WITH}, 
keywordstyle=\color{green!50!black}\bfseries, deletekeywords={TIME, process, NUMBER}, caption={Example \system{} query, which creates a table (\texttt{tableA}) of cars observed by some highway camera \texttt{camA}, then uses it to compute two aggregation results $S_1$ and $S_2$.}, captionpos=b, label={lst:car-query-example}]
/* Select 1 month time window from camera, split video into chunks */
SPLIT camA
    BEGIN 12-01-2020/12:00am END 01-01-2021/12:00am
    BY TIME 5sec STRIDE 0sec
    INTO chunksA;
/* Process chunks using analyst's code, store outputs in tableA */
PROCESS chunksA USING model.py TIMEOUT 1sec
    PRODUCING 10 ROWS
    WITH SCHEMA (plate:STRING="", color:STRING="", speed:NUMBER=0)
    INTO tableA; 
/* S1: average speed of all cars */
SELECT AVG(range(speed, 30, 60)) FROM tableA);
/* S2: count total unique cars of each color */
SELECT color,COUNT(plate) FROM
    (SELECT plate, color FROM tableA)
GROUP BY color WITH KEYS ["RED", "WHITE", "SILVER"];
\end{lstlisting}
\vspace{-22pt}
\end{figure}

\Para{\processcmd{} executables.} 
In addition to the written query, analysts must also attach the executable(s) used by \processcmd{} statements of their query. These executables take a single chunk of video as input and produce a set of rows as output. 
They can maintain arbitrary state and use arbitrary operations (\eg{} custom ML models for CV tasks) while processing a single chunk, but cannot preserve any state across chunks. 
\system{} interprets the output of each chunk according to the \processcmd{} schema and ignores extraneous columns or rows (beyond $\maxrows$). \extendedonly{The analyst is responsible for ensuring the executable output matches the schema.}


\Para{Mitigating side channels.}
To ensure that an event's impact on the output can be bound by its duration, \system{} must ensure that the output of processing some chunk $i$ can \emph{only} be influenced by what is visible in chunk $i$ and not any other chunk $j$. To achieve this, \system{} executes each instance of the \texttt{PROCESS} executable inside an isolated environment (implementation details in \App{isolated-execution}).
To prevent timing side channels, each instantiation must output a value within a fixed time limit, otherwise it is assigned a default output value.
The time limit and default values do not impact the query sensitivity, but they must be fixed to prevent side channels, thus the analyst has the freedom to choose them as part of the \processcmd{}.

\extendedonly{All query statements are executed lazily. A \selectcmd{} only computes a release value once all of the necessary video data exists. For example, if a \splitcmd{} chooses a segment of video from now to one month in the future, but the \selectcmd{} uses a \texttt{GROUPBY} to compute an aggregate value per day, this value can be computed and safely released at the end of each day, rather than waiting for the end of the month.}

\Para{Example query.}
In Listing~\ref{lst:car-query-example}, we provide a sample \system{} query to illustrate the syntax and use of privacy-related aggregation constraints. The full query grammar, language specification, and set of constraints are detailed in~\App{grammar}.



\splitcmd{} selects 1 month of video from \texttt{camA} and splits it into 5-second-long chunks (535,680 chunks total). \processcmd{} uses the executable \texttt{model.py} \extendedonly{(contents in \App{example-process})} to process each chunk. \texttt{model.py} internally employs a custom ML model to detect and track unique cars in the input chunk and outputs the license plate, color, and current speed of up to 10 cars observed. \texttt{plate} and \texttt{color} are arbitrary strings, with a default value of empty string, while \texttt{speed} is a floating-point number with a default value of 0.

The first \selectcmd{} ($S_1$) computes the average speed across all cars. In order to compute a bound on the average of the speed column, its range must be explicitly defined. The \texttt{range} function truncates values to fit within the limits, providing a constraint on how much each row can impact the average (\SecNS{system:sensitivity}).

The second \selectcmd{} ($S_2$) groups the table by color (it is required to specify explicit keys for the \texttt{GROUPBY}, because otherwise the presence of a rare key itself could be used to leak information~\cite{google-dp}), and then \texttt{COUNT}s the number of rows in each group. The ``range'' of these counts (the amount each chunk could contribute) is implicitly [0,10] because \processcmd{} specified $\maxrows{}=10$. Each of the three counts is a separate data release, receives its own sample of noise, and consumes additional privacy budget (\SecNS{budget}).

\Para{Interface limitations.}
The main restriction introduced by \system{}'s interface is the inability to maintain state across separate chunks. However, in most cases this does not preclude queries, it simply requires them to be expressed differently. One broad class of such queries are those that operate over \emph{unique} objects.
Consider the example car counting query above. If a car enters the camera view at chunk $i$ and is last visible in chunk $i+n$, the table will include $n$ rows for the same car instead of the expected 1. Since cars can be uniquely identified by their license plate, this can be solved by adding a \texttt{GROUPBY plate} as an intermediate operator before the final count or average.

Suppose instead the query were counting people, who do not have globally unique identifiers. To handle this, the \processcmd{} executable can only output a row for people that \emph{enter} the scene \emph{during that chunk} (and ignore any people that are already visible at the start of a chunk). This ensures that each appearance of a person corresponds to a single row in the table.
\extendedonly{Explain this example in more detail.}

\system{}'s aggregation interface imposes some limitations beyond traditional SQL (\eg{} the \texttt{SELECT} must specify the range of each column, and must specify \texttt{GROUPBY} keys), but these are equivalent to the limitations of DP SQL interfaces in prior work (full list in \App{grammar}).

\tightsubsection{Query Sensitivity}
\label{sec:system:sensitivity}

The sensitivity of a \system{} query is the maximum amount the final query output could differ given the presence or absence of any \pkbounded{} event in the video. This can be broken down into two questions: (1) what is the maximum number of rows a \pkbounded{} event could impact in the analyst-generated intermediate table, and (2) how much could each of these rows contribute to the aggregate output. We discuss each in turn.

\Para{Contribution of a $\pk$ event to the table.} An event that is visible in even a single frame of a chunk can impact the output of that chunk arbitrarily, but due to \system{}'s isolated execution environment, it can \emph{only} impact the output of that chunk, not any others. Thus, the number of rows a \pkbounded{} event could impact is dependent on the number of chunks it spans 
(an event spans a set of chunks if it is visible in at least one frame of each). 


In the worst case, an event spans the most contiguous chunks when it is first visible in the last frame of a chunk. 
Given a chunk duration $c$ (same units as $\rho$) a single event segment of duration $\rho$ can span at most $\maxchunks{\rho}$ chunks:
\vspace{-5pt}
\begin{equation}
    \label{eq:max-chunks}
    \maxchunks{\rho} = 1 + \lceil\frac{\rho}{c}\rceil
\end{equation}
\vspace{-20pt}

\begin{defn}[Intermediate Table Sensitivity]
    Consider a privacy policy $\pk$, and an intermediate table $t$ (created with a chunk size of $c_t$ and maximum per-chunk rows $\maxrows_t$). The \emph{sensitivity} of $t$ w.r.t $\pk$, denoted $\Delta_{\pk}$, is the maximum number of rows that could differ given the presence or absence of any \pkbounded{} event:
\begin{equation}
    \label{eq:base-table-sensitivity}
    \Delta_{(\rho, K)}(t) \le \maxrows{}_t \cdot K \cdot \maxchunks{\rho}
\end{equation}
\end{defn}


\begin{proof}
In the worst case, none of the $K$ segments overlap, and each starts at the last frame of a chunk. Thus, each spans a separate $\maxchunks{\rho}$ chunks (Eq.~\ref{eq:max-chunks}). For each of these chunks, all of the $\maxrows$ output rows could be impacted.
\end{proof}
\vspace{-3pt}

\extendedonly{%
Consequently, since the ceiling operation ``round''s the duration given by a policy up to the nearest chunk, \system{} actually protects $(\hat{\rho},K)$-bounded events with $\hat{\rho} = \maxchunks{\rho} \cdot c_t$ seconds, which is always greater than $\rho$.
}

\Para{Sensitivity propagation for \pkbounded{} events.}
Prior work~\cite{flex, privatesql} has shown how to compute the sensitivity of a SQL query over \emph{traditional} tables. Assuming that queries are expressed in relational algebra, they define the sensitivity recursively on the abstract syntax tree.
Beginning with the maximum number of rows an individual could influence in the input table, they provide rules for how the influence of an individual propagates through each relational operator and ultimately impacts the aggregation function. 


While we can leverage the same high-level approach of propagating sensitivity recursively, the semantics of the intermediate table are unique from prior work and thus require careful consideration to ensure the privacy definition is rigorously guaranteed. In this work, we determined the set of operations that can be enabled over \system{}’s intermediate tables, derived the corresponding sensitivity rules, and proved their correctness. Many rules end up being analogous or similar to those in prior work, but \texttt{JOIN}s are different. We provide a brief intuition for these differences below. Fig.~\ref{fig:sensitivity-table} in \App{sensitivity} contains the complete algorithm for determining the sensitivity of a \system{} query (Theorem~\ref{thm:single-query}). We provide a proof of Theorem~\ref{thm:single-query} in \App{proofs}.

\Para{Privacy semantics of untrusted tables.} As an example, consider a query that computes the size of the intersection between two cameras, \texttt{PROCESS}'d into intermediate tables $t_1$ and $t_2$ respectively. If $\Delta(t_1) = x$ and $\Delta(t_2) = y$, it is tempting to assume $\Delta(t_1 \cap t_2) = \min(x,y)$, because a value needs to appear in both $t_1$ and $t_2$ to appear in the intersection. However, because the analyst’s executable can populate the table arbitrarily, they can ``prime'' $t_1$ with values that would only appear in $t_2$, and vice versa. As a result, a value need only appear in either $t_1$ or $t_2$ to show up in the intersection, and thus $\Delta(t_1 \cap t_2) = x + y$.



\begin{theorem}
    \label{thm:single-query}
    \system{}'s sensitivity computation algorithm (Fig.~\ref{fig:sensitivity-table}, \App{sensitivity}) provides \pkeprivacy{} for a single query $Q$ over video $V$.
\end{theorem}

\subsection{Handling Multiple Queries}
\label{sec:budget}

In traditional DP, the parameter $\epsilon$ is viewed as a ``privacy budget''. Informally, $\epsilon$ defines the total amount of information that may be released about a database, and each query consumes a portion of this budget. Once the budget is depleted, no further queries can be answered.

Rather than assigning a single global budget to an entire video, \system{} allocates a separate budget of $\epsilon$ to each frame of a video.
When \system{} receives a query $Q$ over frames $[a,b]$ requesting budget $\epsilon_Q$, it only accepts the query if \emph{all} frames in the interval $[a-\rho, b+\rho]$ have sufficient budget $\ge \epsilon_Q$, otherwise the query is denied (Alg.~\ref{alg:main} Lines 1-3).
If the query is accepted, \system{} then subtracts $\epsilon_Q$ from each frame in $[a,b]$, but \emph{not} the $\rho$ margin (Alg.~\ref{alg:main} Lines 4-5). We require sufficient budget at the $\rho$ margin to ensure that any single segment of an event (which has duration at most $\rho$) cannot span two temporally disjoint queries (\App{proofs}).

Note that since each \selectcmd{} in a query represents a separate data release, the total budget $\epsilon_Q$ used by a query is the sum of the $\epsilon_i$ used by each of the $i$ \selectcmd{}s. The analyst can specify the amount of budget they would like to use for each release (via the \texttt{CONSUMING} directive in a \texttt{SELECT}, \App{grammar}). 
Revisiting our example query from \Sec{system:query}, if $S_1$ requested $\frac{1}{6}$ units for each color count (3 keys, total $\frac{1}{2}$), and $S_2$ requested $\frac{1}{2}$ units, then the privacy budget for the entire query would be $\epsilon_Q = 1$.

\extendedonly{The amount of privacy budget remaining per frame is public information because it is only a function of past queries, not of the data. If analysts wish to query the interval $[a,b]$ but only part of the range has sufficient budget, they can adjust their interval to only consider frames with sufficient budget.}

\Para{Putting it all together.} 
Algorithm~\ref{alg:main} presents a simplified (single source video and aggregation) version of the \system{} query execution process. \extendedonly{(for full algorithm see \App{full-algorithm})}

\begin{algorithm}
\small
\SetKwInOut{Input}{Input}
\SetKwInOut{Output}{Output}

\Input{
\system{} query $Q = \{c, schema, S, \epsilon_Q\}$, video $V$, interval $I = begin:end$, policy $(\rho, K, \epsilon)$}
\Output{Query answer $A$}

\ForEach{frame $f \in V[I\pm\rho]$}{
    \If{$f.budget < \epsilon_Q$}{
        \Return \texttt{DENY}
    }
}
\ForEach{frame $f \in V[I]$}{
    $f.budget$ \texttt{-=} $\epsilon_Q$
}

\texttt{chunks} $\gets$ Split $V[I]$ into chunks of duration $c$

$T \gets \texttt{Table(schema)}$

\ForEach{$chunk \in chunks$}{
    \texttt{rows} $\gets F(chunk)$ \tcp{in isolated environment}
    $T$\texttt{.append}(rows)
}

$r \gets$ execute SQL query $S$ over table $T$
    
$\Delta_{(\rho,K)} \gets $ compute recursively over the structure of $S$ (\SecNS{system:sensitivity})
    
$\eta \gets Laplace(\mu=0, b=\frac{\Delta}{\epsilon_Q})$

$A \gets r + \eta$

\caption{\system{} Query Execution (simplified)}
\label{alg:main}
\end{algorithm}

\begin{theorem}
\label{thm:multiple-queries}
Consider an adaptive sequence (\SecNS{threat-model}) of $n$ queries $Q_1,\ldots,Q_n$, each over the same camera $C$, a privacy policy $(\rho_C, K_C)$, and global budget $\epsilon_C$. \system{} (Algorithm~\ref{alg:main}) provides $(\rho_C, K_C, \epsilon_C)$-privacy for all $Q_1,\ldots,Q_n$. 
\end{theorem}
\noindent We provide the full proof in \App{proofs}.


\tightsection{Query Utility Optimization}
\label{sec:practical}

The noise that \system{} adds to a query result is proportional to both the privacy policy $\pk$ and the range of the aggregated values (the larger the range, the more noise \system{} must add to compensate for it). In this section we introduce two optional optimizations that \system{} offers analysts to improve query accuracy while maintaining an equivalent level of privacy: one reduces the $\rho$ needed to preserve privacy (\SecNS{practical:mask}), while the other reduces the range for aggregation (\SecNS{practical:split}).



\tightsubsection{Spatial Masking}
\label{sec:practical:mask}
\vspace{-5pt}

\begin{figure}[t!]
  \centering
      \begin{subfigure}[b]{.46\textwidth}
        \includegraphics[width=\linewidth]{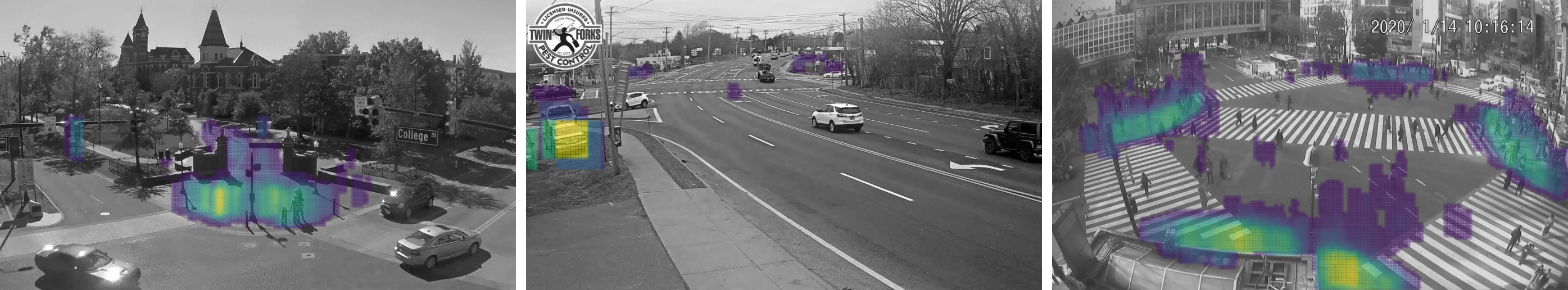}
      \end{subfigure}
    \begin{subfigure}[b]{0.15\textwidth}
        \begin{subfigure}[b]{\linewidth}
            \includegraphics[width=\linewidth]{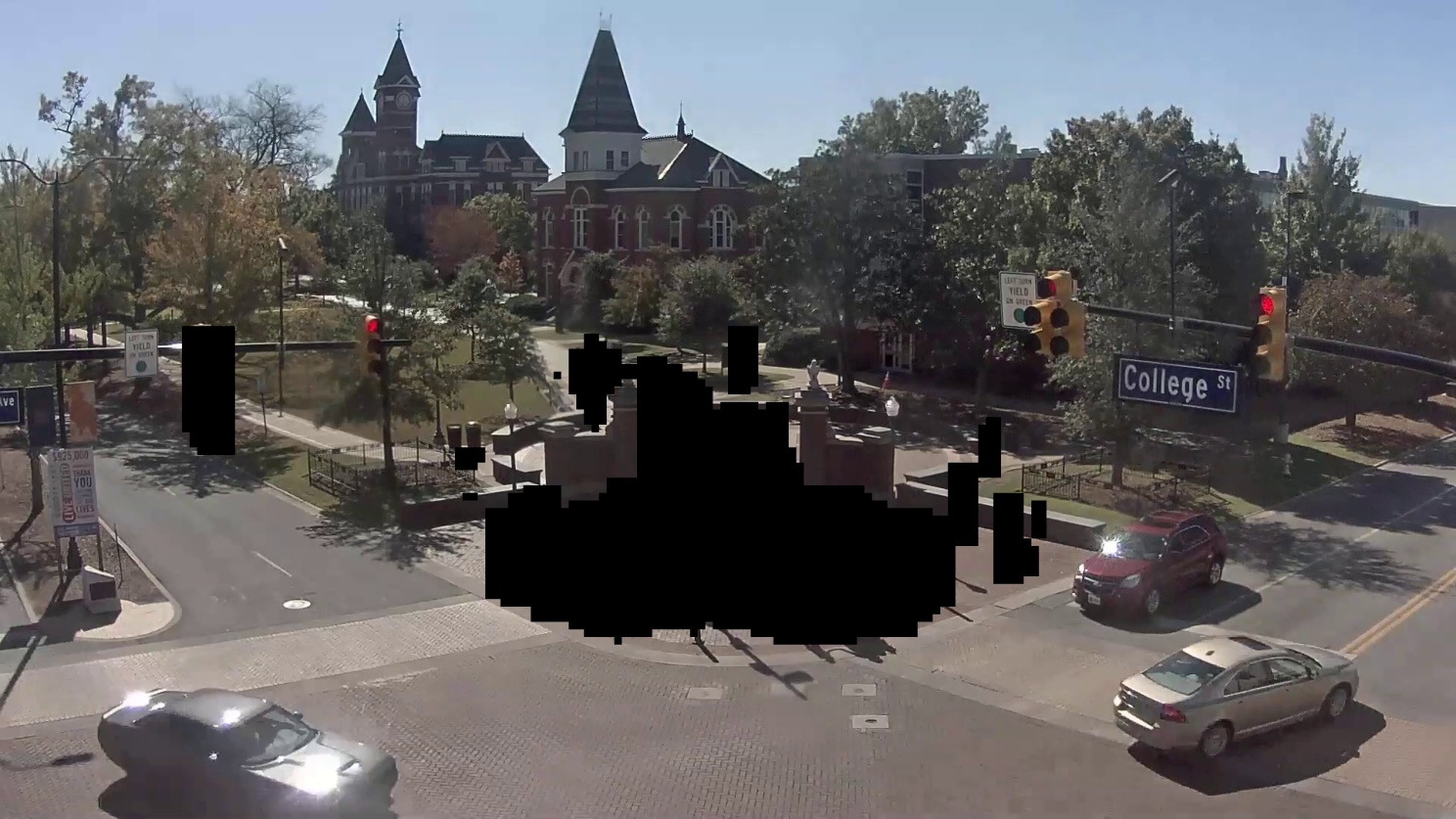}
            \caption{\auburn{}}
            \label{fig:mask:auburn}
        \end{subfigure}
    \end{subfigure}
    \begin{subfigure}[b]{0.15\textwidth}
        \begin{subfigure}[b]{\linewidth}
            \includegraphics[width=\linewidth]{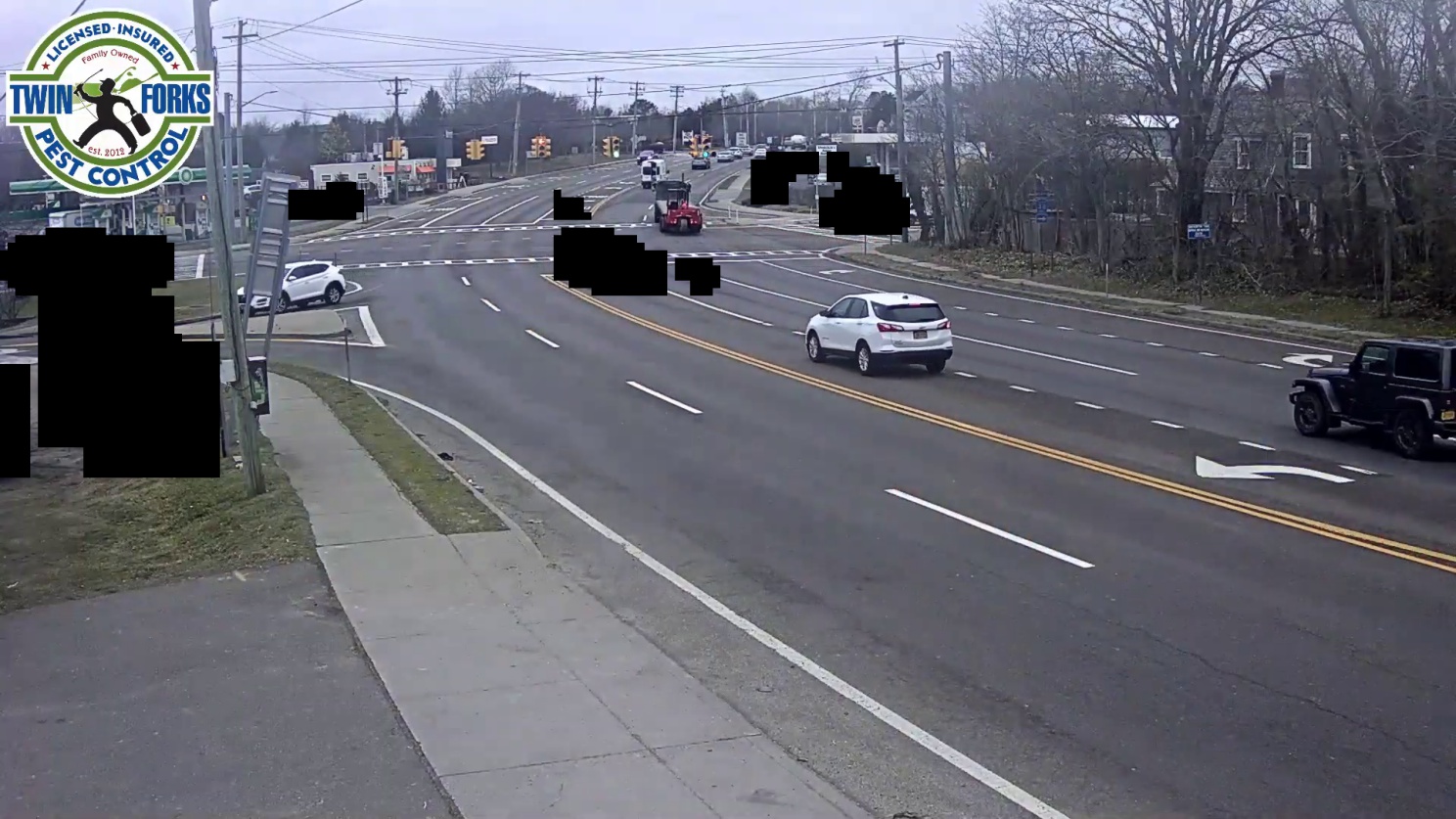}
            \caption{\hampton{}}
            \label{fig:mask:hampton}
        \end{subfigure}
    \end{subfigure}
    \begin{subfigure}[b]{0.15\textwidth}
        \begin{subfigure}[b]{\linewidth}
            \includegraphics[width=\linewidth]{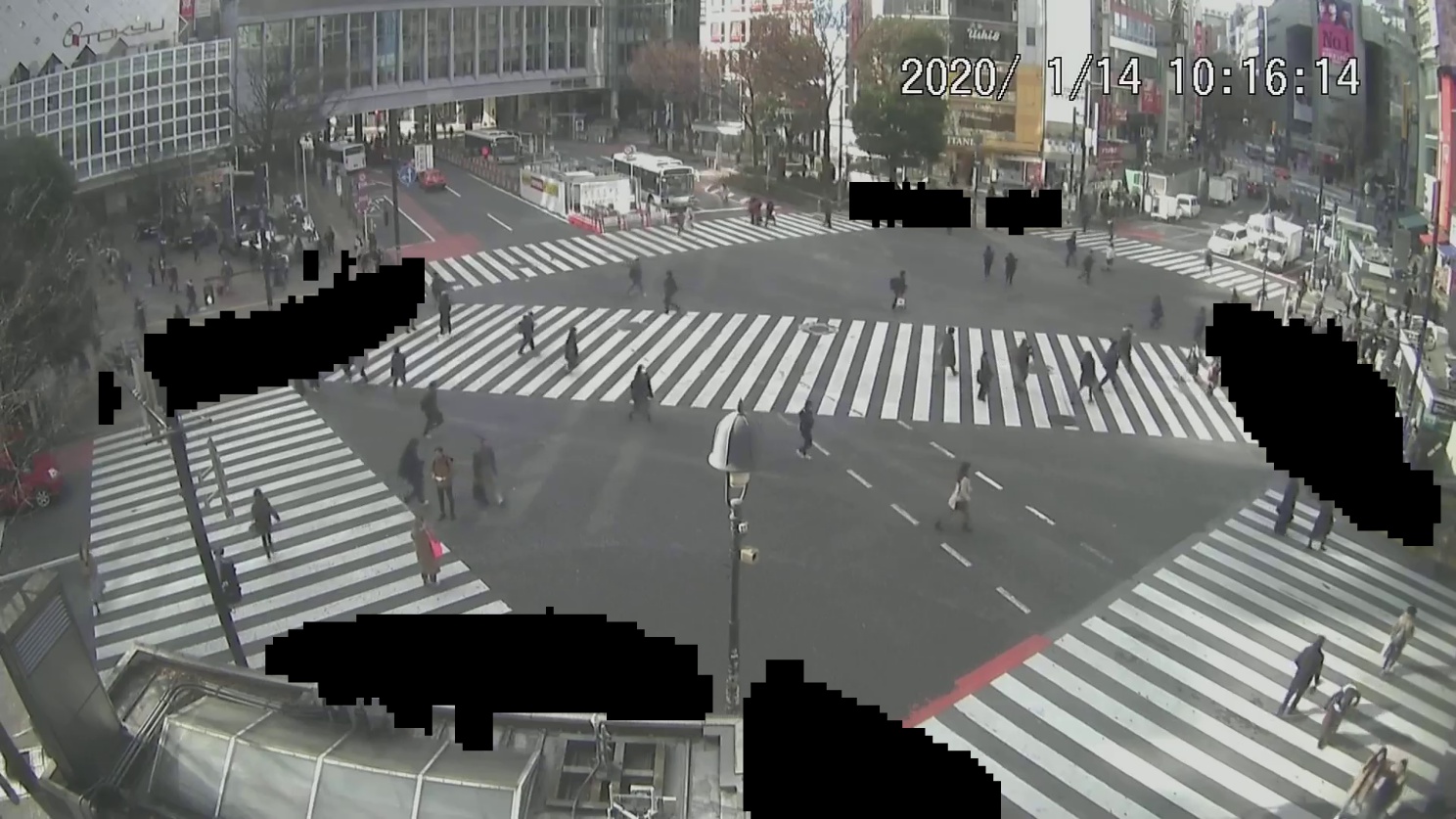}
            \caption{\shibuya{}}
            \label{fig:mask:shibuya}
        \end{subfigure}
    \end{subfigure}
    \vspace{2pt}
    \tightcaption{Heatmaps (yellow/blue indicates max/min persistence) and resulting masks for each video in our dataset; persistence range is normalized per video.} 
    \label{fig:heatmap}
\end{figure}

\begin{figure*}
  \centering
    \begin{subfigure}[b]{0.31\textwidth}
        \includegraphics[width=\linewidth]{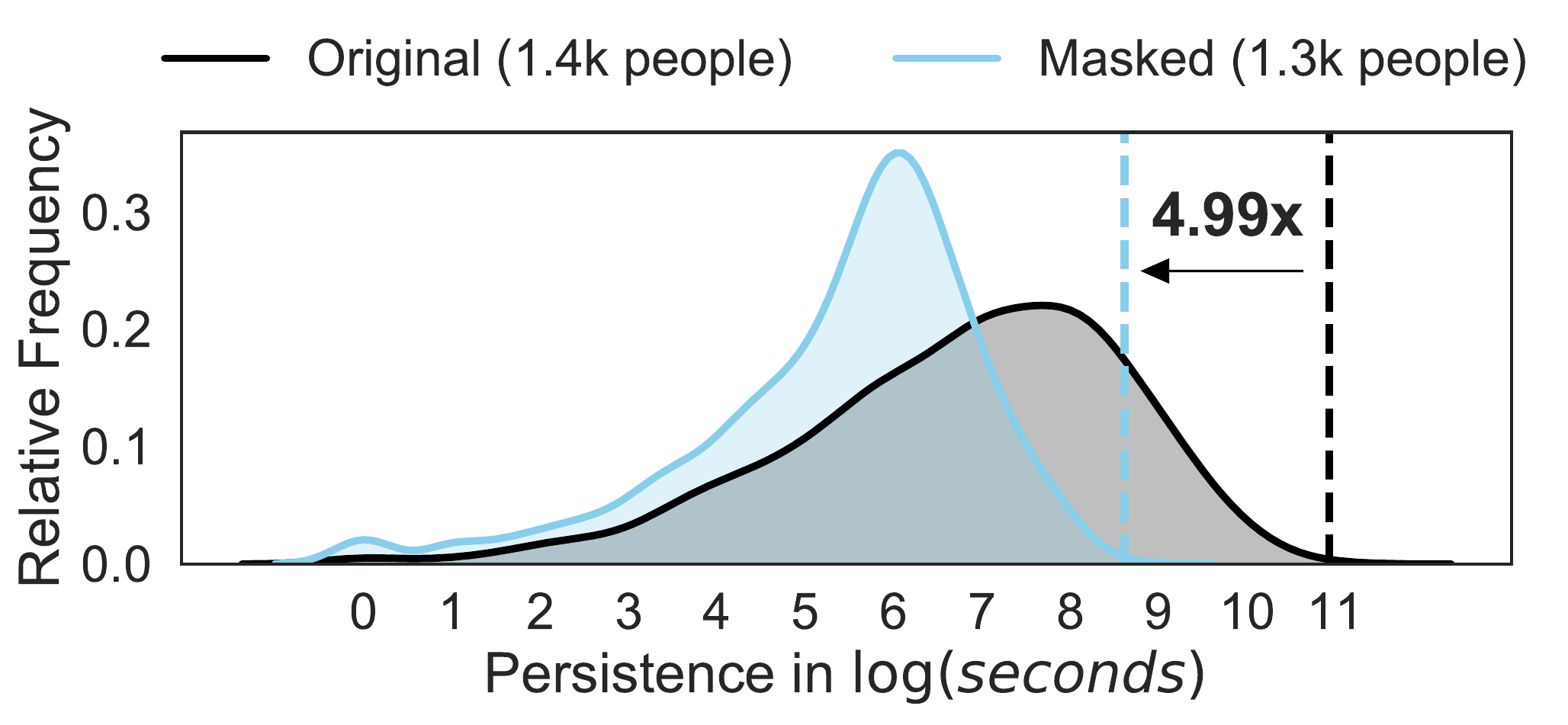}
        \caption{\auburn{}}
        \label{fig:screenshots:auburn}
    \end{subfigure}
    \begin{subfigure}[b]{0.31\textwidth}
        \includegraphics[width=\linewidth]{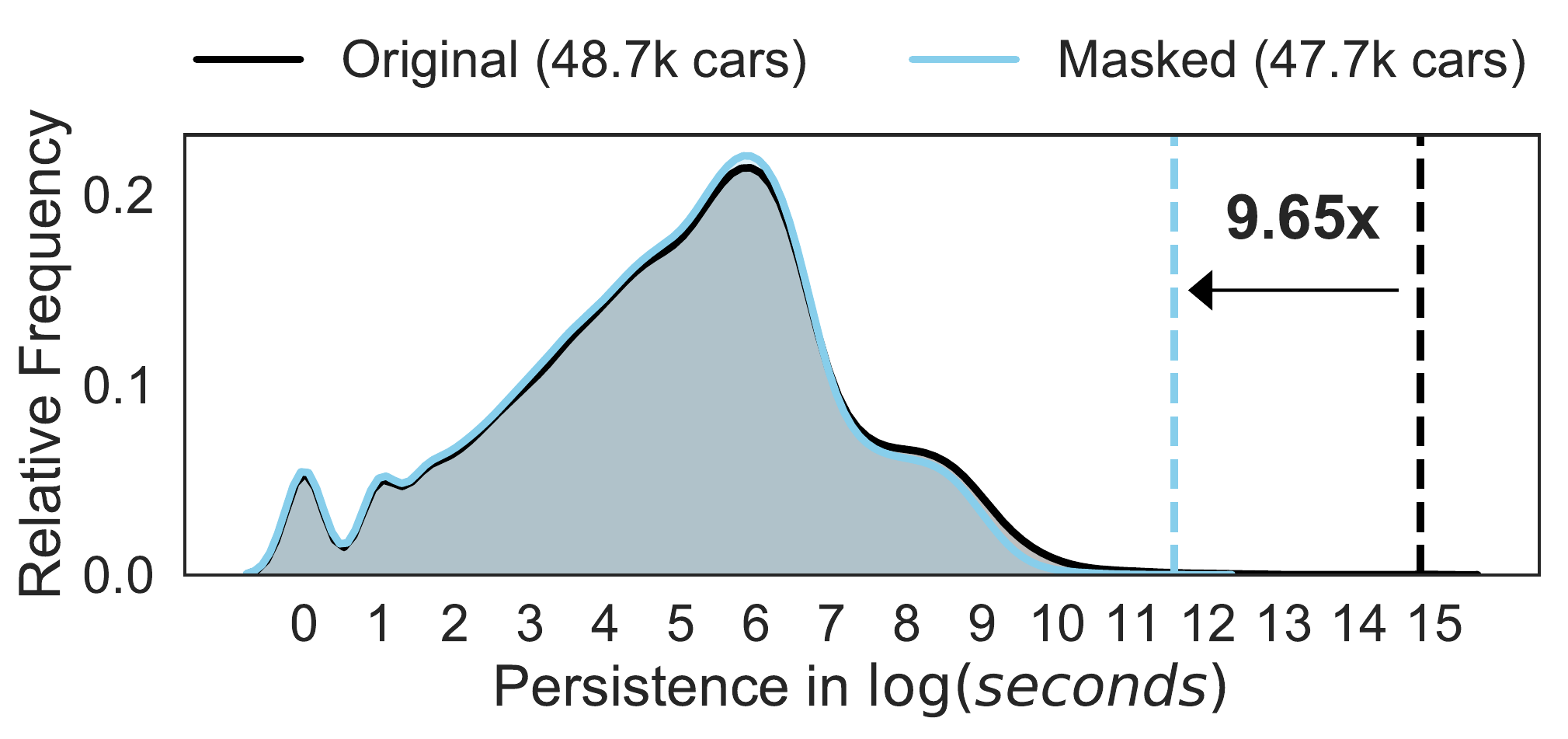}
        \caption{\hampton{}}
        \label{fig:screenshots:hampton}
    \end{subfigure}
    \begin{subfigure}[b]{0.31\textwidth}
        \includegraphics[width=\linewidth]{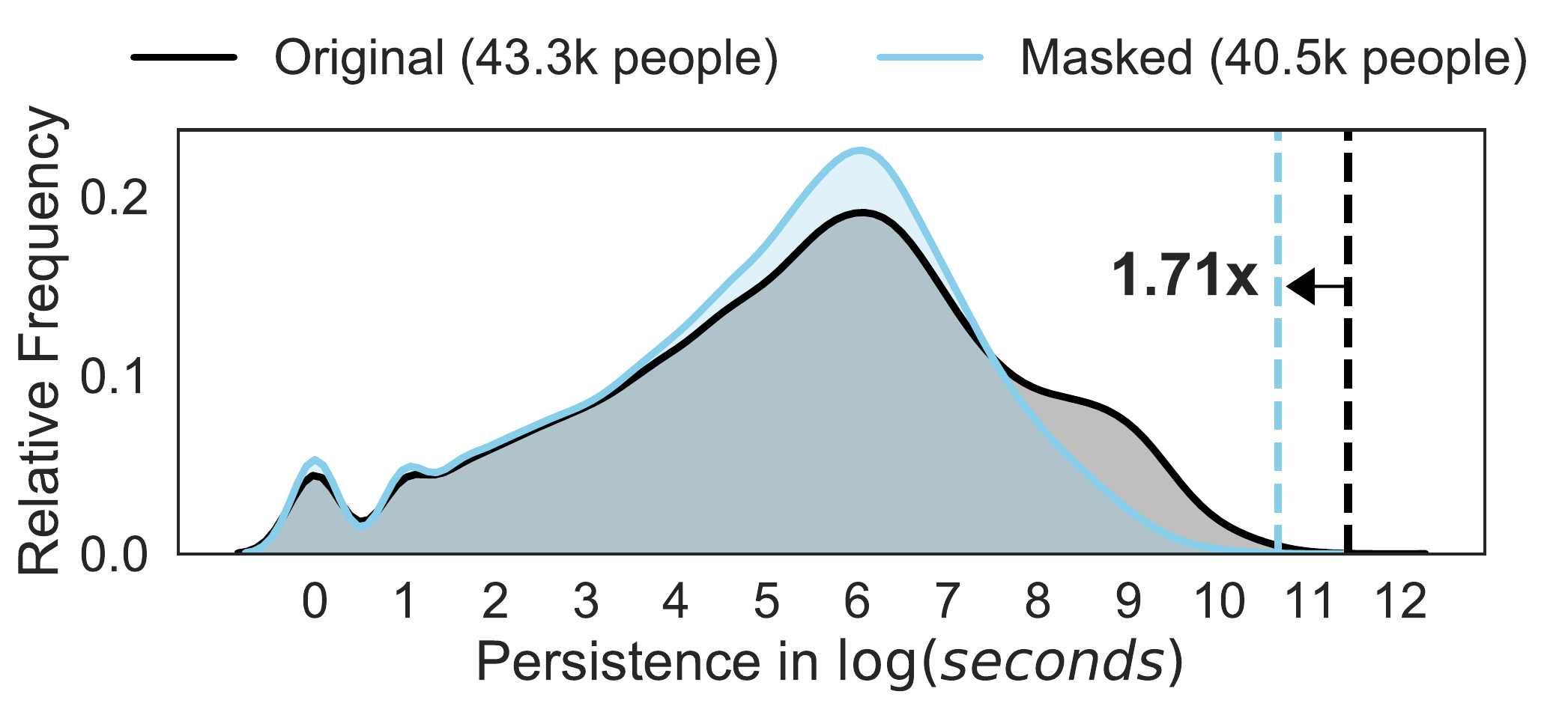}
        \caption{\shibuya{}}
        \label{fig:screenshots:shibuya}
    \end{subfigure}
    \vspace{2pt}
    \tightcaption{The distribution of private objects' durations (persistence) is heavy tailed. Applying the mask from~\Fig{heatmap} significantly lowers the maximum duration, while still allowing most private objects to be detected. The key denotes the total number of private objects detectable before and after applying the mask. The dotted lines highlight the maximum persistence, and the arrow text denotes the relative reduction. 
}
    \label{fig:masking:distributions}
\end{figure*}

\Para{Observation.} 
In certain settings, a few individuals may be visible to a camera for far longer than others (\eg{} those sitting on a bench or in a car), creating a heavy-tailed distribution of presence durations. \Fig{heatmap} (top row) provides some representative examples. Setting $\pk$ to the maximum duration in such distributions would result in a large amount of noise needed to protect just those few individuals; all others could have been protected with a far lower amount of noise. We observe that, in many cases, lingering individuals tend to spend the majority of their time in one of a few fixed regions in the scene, but a relatively short time in the rest of the scene. 
For example, a car may be parked in a spot for hours, but only visible for 1 minute while entering/leaving the spot.

\Para{Opportunity.}
Masking regions in the scene that house lingering individuals (\ie{} making them not visible to the analyst's video processing) would drastically reduce the \emph{observable} maximum duration of individuals' presence, e.g., the parked car from above would be observable for 1 min rather than hours. This, in turn, would result in a lower $\rho$ in the privacy policy, but an equivalent level of privacy--all individuals would still be protected. Of course, this technique is only useful to an analyst when the remaining (unmasked) part of the scene includes all the information needed for the query at hand, e.g., if counting cars, masked cars should be visible once they start moving.


\Para{Optimization.}
At camera-registration time, instead of choosing a single $\pk$ policy per camera, the video owner can instead release a map of potential masks to corresponding $\pk$ policies that would provide an equivalent level of privacy. 
Then, at query time, the analyst can (optionally) choose the mask that would minimally impact their query processing while maximizing the level of noise reduction (via the lower $\pk$ bound). If a mask is chosen, \system{} applies it to all frames of each chunk before passing it to the analyst's video processing executable (\SecNS{system:query}), and uses the adjusted $\pk$ in the sensitivity calculation (\SecNS{system:sensitivity}).

To generate the corresponding policy for a mask, the video owner can apply the same process of analyzing prior video from the camera (this time, with the mask applied first) to estimate the maximum \emph{observable} duration of objects they wish to protect. However, the space of all possible masks is infeasible to enumerate, and computing the policy for each is computationally expensive. Fortunately, the set of masks that would actually reduce the observable duration is relatively small and query-agnostic; we provide a computationally-efficient algorithm to find this minimal set of masks (and compute their policies) in \App{mask-data-structure}.


\Para{Noise reduction.}
We demonstrate the potential benefit of masking on three queries (Q1-Q3) from our evaluation (Table~\ref{tab:query-results}). Given the query tasks (counting unique people and cars), we chose masks that would maximally reduce $\rho$ without impacting the object counts; the bottom row of \Fig{heatmap} visualizes our masks. \Fig{masking:distributions} shows that these masks reduce maximum durations by 1.71-9.65$\times$. We extend this evaluation to 7 more videos from BlazeIt~\cite{blazeit} and MIRIS~\cite{miris} in \App{masking-effectiveness}.


\Para{Masking vs. denaturing.}
Although masking is a form of denaturing, \system{} uses it differently than
the prior approaches in~\Sec{prior:denaturing}, in order to sidestep their issues.
Rather than attempting to dynamically hide individuals as they move through the scene,
\system{}'s masks cover a \emph{fixed} location in the scene and are publicly
available so analysts can account for them in their query implementation.
Also, masks are used as an optional modification to the input video; the rest of the \system{} pipeline, and thus its formal privacy guarantees,
remain the same.


\begin{table}
    \small
    \begin{tabular}{lllllll}
        \textbf{Video} & \textbf{Max(frame)} & \textbf{Max(region)} & \textbf{Reduction} \\ \hline
        \auburn{}  & 3 & 6 & 2.00$\times$ \\
        \hampton{} & 40 & 23 & 1.74$\times$ \\
        \shibuya{} & 37 & 16 & 2.25$\times$ \\
    \end{tabular}
    \vspace{5pt}
    \tightcaption{Reduction in max output range from splitting each video into distinct regions.
    Reduction shows the factor by which the noise could be reduced. 
    2$\times$ cuts the necessary privacy level in half.
    \vspace{-5pt}
    }
    \label{tab:spatial-split}
\end{table}

\tightsubsection{Spatial Splitting}
\label{sec:practical:split}
\vspace{-5pt}

\Para{Observation.}
(1) At any point in time, each object typically occupies a relatively small area of a video frame.
(2) Many common queries (e.g., object detections) do not need to examine the entire contents of a frame at once, i.e., if the video is split spatially into regions, they can compute the same total result by processing each of the regions separately. 

\Para{Opportunity.}
\system{} already splits videos temporally into chunks. If each chunk is further divided into spatial regions and an individual can only appear in one of these chunks at a time, then their presence occupies a relatively smaller portion of the intermediate table (and thus requires less noise to protect). Additionally, the maximum duration of each individual region may be smaller than the frame as a whole.

\Para{Optimization.}
At camera-registration time, \system{} allows video owners to manually specify boundaries for dividing the scene into regions. They must also specify whether the boundaries are soft (individuals may cross them over time, \eg{} between two crosswalks) or hard (individuals will never cross them, \eg{} between opposite directions on a highway). 
At query time, analysts can optionally choose to spatially split the video using these boundaries. 
Note that this is in addition to, rather than in replacement of, the temporal splitting.
If the boundaries are soft, tables created using that split must use a chunk size of 1 to ensure that an individual can always be in at most 1 chunk. If the boundaries are hard, there are no restrictions on chunk size since the video owner has stated the constraint will always be true.

\Para{Noise reduction.}
We demonstrate the potential benefit of spatial splitting on three videos from our evaluation (Q1-Q3). For each video, we manually chose intuitive regions: a separate region for each crosswalk in \auburn{} and \shibuya{} (2 and 4, respectively), and a separate region for each direction of the road in \hampton{}. Table~\ref{tab:spatial-split} compares the range necessary to capture all objects that appear within one chunk in the entire frame compared to the individual regions. The difference (1.74-2.25$\times$) represents the potential noise reductions from splitting: noise is proportional to $\max(frame)$ or $\max(region)$ when splitting is disabled or enabled, respectively. 

\Para{Grid Split.}
To increase the applicability of spatial splitting, \system{} could allow analysts to divide each frame into a grid and remove the restrictions on soft boundaries to allow any chunk size. This would require additional estimates about the max size of any private object (dictating the max number of cells they could occupy at any time), and the maximum speed of any object across the frame (dictating the max number of cells they could move between). We leave this to future work.

\tightsection{Evaluation}
\label{sec:eval}

The evaluation highlights of \system{} are as follows:
\begin{enumerate}
    \item \system{} supports a diverse range of video analytics queries, including object counting, duration queries, and composite queries; for each, \system{} achieves accuracy within 79-99\% of a non-private system, while protecting all individuals with \pkeprivacy{} (\S\ref{eval:case-studies}). 
    \item \system{} enables video owners and analysts to flexibly and formally trade utility loss and query granularity while preserving the same privacy guarantee (\SecNS{eval:param-sweep}).
\end{enumerate}

\tightsubsection{Evaluation Setup}
\label{sec:dataset}
\vspace{-2pt}

\Para{Datasets.}
We evaluated \system{} primarily using three representative video streams (\auburn{}, \hampton{} and \shibuya{}, screenshots in~\Fig{heatmap}) that we collected from YouTube spanning 12 hours each (6am-6pm). 
For one case study (multi-camera), we use the Porto Taxi dataset~\cite{porto} containing 1.7mil trajectories of all 442 taxis running in the city of Porto, Portugal from Jan. 2013 to July 2014. We apply the same processing as \cite{rexcam-hotmobile} to emulate a city-wide camera dataset; the result is the set of timestamps each taxi would have been visible to each of 105 cameras over the 1.5 year period.

\Para{Implementation.}
We implemented \system{} in 4k lines of Python.
All analyst-provided \texttt{PROCESS} executables and camera-owner $\pk$ estimation use
the Faster-RCNN~\cite{faster_rcnn} model in Detectron-v2~\cite{detectron_v2} for object detection, and DeepSORT~\cite{deepsort} for object tracking. 
For these models to work reasonably given the diverse content of the videos, we chose the hyperparameters for detection and tracking on a per-video basis; 
more details are presented in \App{cv-params}.

\Para{Privacy Policies.}
We assume the video owner's underlying privacy goal is to ``protect the appearance of all individuals''.
For each camera, we use the strategy in~\Sec{practical:mask}, analyzing past video with CV algorithms to create a map between masks and $\pk$ policies that achieve this goal.

\Para{Query Parameters.}
For each query, we first chose a mask that covered as much area as possible (to get the minimal $\rho$) without disrupting the query. The resulting $\rho$ values are in Table~\ref{tab:query-results}.
We use a budget of $\epsilon=1$ for each query. 
We chose query windows sizes ($W$), chunk durations ($c$), and column ranges to best approximate the analyst's expectations for each query (as opposed to picking optimal values based on a parameter sweep, which the analyst is unable to do).
We empirically explore the impact of each parameter value in \Sec{eval:param-sweep}.


\Para{Baselines.}
For each query, we compute \texttt{accuracy} by comparing the output of \system{} (impacted by both chunking and addition of noise to the aggregate) to running the same exact query implementation without \system{} (\ie{} without chunking or noise). Since \system{}'s results include a component of random noise, we execute each query 1000 times, and report the mean accuracy value $\pm$ 1 standard deviation.

\tightsubsection{Query Case Studies}
\label{eval:case-studies}

We formulate five types of queries to span a variety of axes
(target object class, number of cameras, aggregation type, query duration, standing vs. one-off query). 
Fig.~\ref{fig:eval:graph2} displays results for Q1-Q3. Table~\ref{tab:query-results} summarizes the remaining queries (Q4-Q13).

%
%
%

\begin{table*}
\small
\center
\resizebox{\textwidth}{!}{%
\begin{tabular}{|l|l|l|l|l|l|l|l|}
\hline
\textbf{Case \#} & \textbf{Q\#} & \textbf{Query Description} & \textbf{Query Parameters} & \textbf{Video} & \textbf{$\rho$} & \textbf{Query Output} & \textbf{Accuracy} \\ \hline

Case 2 & Q4 & \begin{tabular}[c]{@{}l@{}}Average Taxi Driver Working Hours \\ (\texttt{union} across 2 cameras)\end{tabular}                         & \begin{tabular}[c]{@{}l@{}}$|W| = 365~\text{days}$, $c = 15~\text{sec}$,\\ Agg = \texttt{avg}, range $ = (0, 16)$\end{tabular}                   & \texttt{porto10}, \texttt{porto27}      & [45, 195] sec        & 5.87 hrs              & $94.14\% \pm 0.18\%$  \\ \hline
Case 2 & Q5 & \begin{tabular}[c]{@{}l@{}}Average \# Taxis Traversing 2 Locations on \\ Same Day (\texttt{intersection} across 2 cameras)\end{tabular} & \begin{tabular}[c]{@{}l@{}}$|W| = 365~\text{days}$, $c = 15~\text{sec}$,\\ Agg = \texttt{avg}, range $ = (0, 300)$\end{tabular}                  & \texttt{porto10}, \texttt{porto27}      & [45, 195] sec        & 131 taxis             & $99.80\% \pm 0.13\%$  \\ \hline
Case 2 & Q6 & \begin{tabular}[c]{@{}l@{}}Identifying Camera with Highest Daily Traffic\\ (\texttt{argmax} across all 105 cameras)\end{tabular}   &  \begin{tabular}[c]{@{}l@{}}$|W| = 365~\text{days}$, $c = 15~\text{sec}$,\\ Agg = \texttt{argmax}\end{tabular}                                                                                                                                             & \texttt{porto0}, ..., \texttt{porto104} & [15, 525] sec        & \texttt{porto20}                & 100.00\%              \\ \hline
\multirow{3}{*}{Case 3} & Q7 &
\multirow{3}{*}{Fraction of trees with leaves (\%)}                                                                            & \multirow{3}{*}{\begin{tabular}[c]{@{}l@{}}$|W| = 12~\text{hrs}$, $c = 1~\text{frame}$, \\ Agg = \texttt{avg}, range $ = (0,100)$\end{tabular}} & \auburn{}             & 48.89 sec            & 15/15 = 1.00          & $99.90\% \pm 0.11\%$  \\ \cline{5-8} 
                                                                              &   Q8                                              &    &                                                                                                                                             & \hampton{}            & 6.21 min             & 3/7 = 0.43            & $98.24\% \pm 1.90\%$  \\ \cline{5-8} 
                                                                           &         Q9                                           &  &                                                                                                                                               & \shibuya{}            & 3.34 min             & 4/6 = 0.67            & $99.39\% \pm 0.66\%$  \\ \hline
\multirow{3}{*}{Case 4} & Q10 &                                                                                                                                
\multirow{3}{*}{Duration of Red Light (seconds)}                                                                               & \multirow{3}{*}{\begin{tabular}[c]{@{}l@{}}$|W| = 12~\text{hrs}$, $c = 10~\text{min}$, \\ Agg = \texttt{avg}, range $ = (0,300)$\end{tabular}}  & \auburn{}             & 0 sec                & 75 sec                   & 100.00\%              \\ \cline{5-8} & Q11 & 
                                                                                                                               &                                                                                                                                                 & \hampton{}            & 0 sec                & 50 sec                   & 100.00\%              \\ \cline{5-8} & Q12 &  
                                                                                                                               &                                                                                                                                                 & \shibuya{}            & 0 sec                & 100 sec                  & 100.00\%              \\ \hline
Case 5 & Q13 &                                                                                                                               
\begin{tabular}[c]{@{}l@{}}\# Unique People (Filter: trajectory \\ moving towards campus)\end{tabular}                         & \begin{tabular}[c]{@{}l@{}}$|W| = 12~\text{hrs}$, $c = 10~\text{min}$,\\ Agg = \texttt{sum}, range $ = (0,25)$\end{tabular}                      & \auburn{}             & 49 sec               & 576 people                   & $79.06\% \pm 4.75\%$ \\ \hline
\end{tabular}}
\vspace{15pt}
\tightcaption{Summary of query results for Q4-Q13. For Case 3 and 5, we use the same masks (and thus $\rho$) from Figure~\ref{fig:heatmap}. For Case 4, we mask all pixels except the traffic light to attain $\rho = 0$. For Case 2 we do not use any masks.}
\vspace{3pt}
\label{tab:query-results}
\end{table*}

\begin{figure*}
  \centering
  \includegraphics[width=0.9\textwidth]{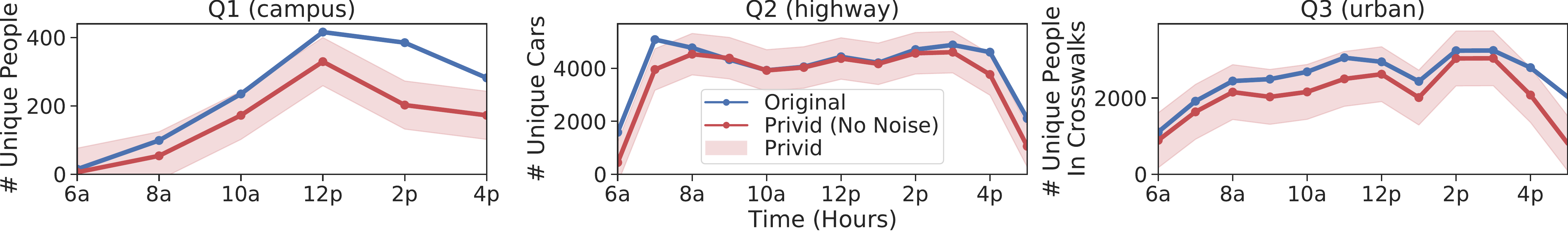}
  \vspace{0.15cm}
  \tightcaption{Time series of \system{}'s output for Case 1 queries. ``Original'' is the baseline query output without using  \system{}. ``Privid (No Noise)'' shows the raw output of \system{} before noise is added. The final noisy output will fall within the range of the red ribbon 99\% of the time.}
  \vspace{3pt}
  \label{fig:eval:graph2}
\end{figure*}

\begin{figure*}
  \centering
  \includegraphics[width=0.9\textwidth]{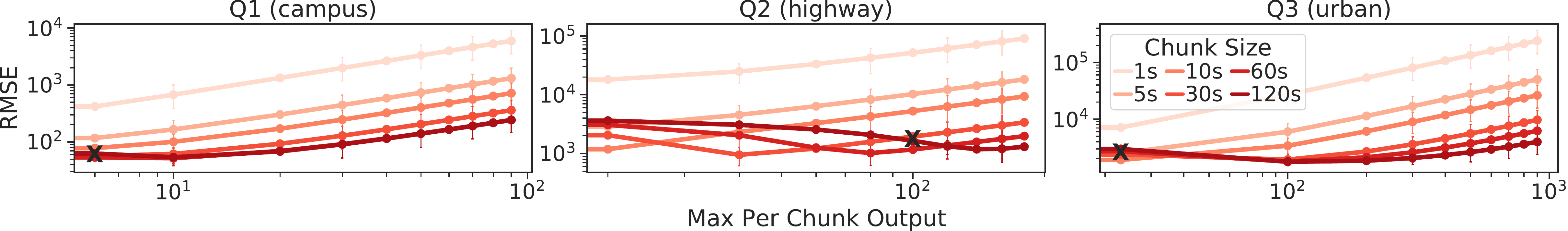}
  \vspace{0.15cm}
  \tightcaption{Impact of chunk size and output range (fixed $l=0$, increasing $u$) on \system{}'s root mean square error (RMSE) for the queries in~\Fig{eval:graph2}. The reference value is the same as~\Fig{eval:graph2}, namely the ``Original'' line. Error bars computed over 100 samples of noisy outputs from \system{}. The ``X'' represents the exact pair of parameters we chose for each video in~\Fig{eval:graph2}. 
    }
  \label{fig:eval:graph3}
\end{figure*}

\begin{figure}
  \centering
  \includegraphics[width=0.85\columnwidth]{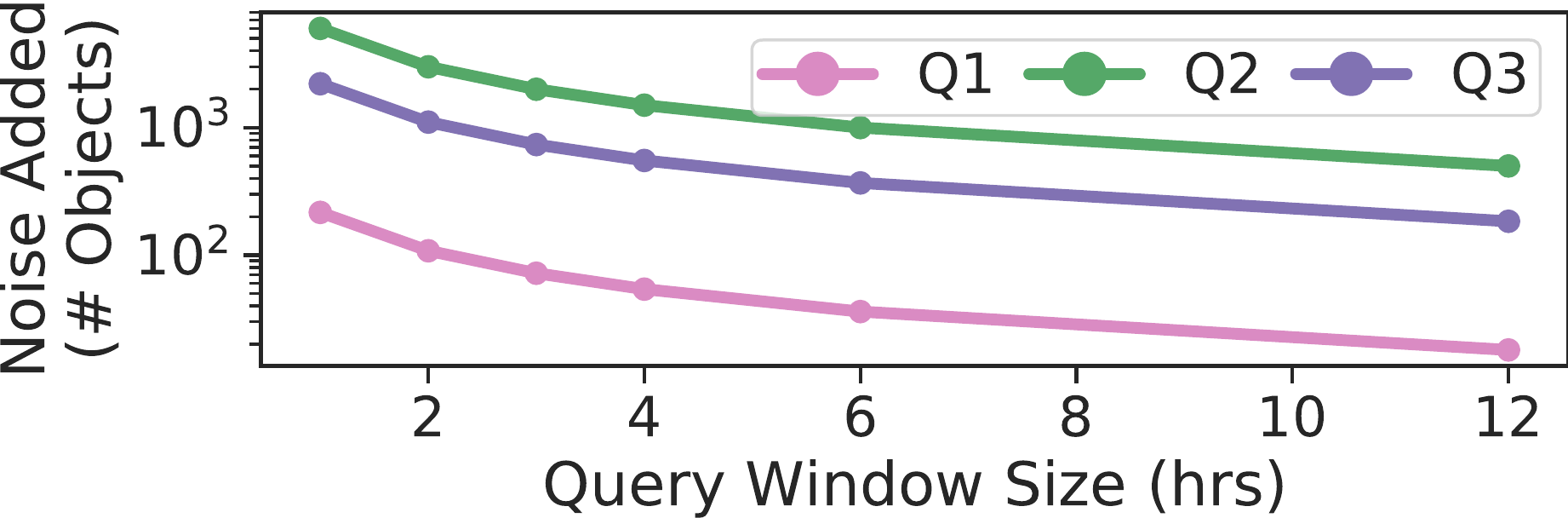} 
  \vspace{0.15cm}
    \tightcaption{Impact of query window size on the amount of noise \system{} must add to meet the privacy guarantee for Q1-Q3.
    }
  \label{fig:eval:graph3b}
\end{figure}

\Para{Case~1: Q1-Q3 (Counting private objects over time)}.
To demonstrate \system{}'s support for standing queries and short (1 hour) aggregation durations, we \texttt{SUM} the number of \textit{unique} objects observed \emph{each hour} over the 12 hours. 

\Para{Case~2: Q4-Q6 (Aggregating over multiple cameras with complex operators)}. 
We utilize \texttt{UNION}, \texttt{JOIN}, and \texttt{ARGMAX} to aggregate over cameras in the \texttt{Porto Taxi Dataset}. 
Due to the large aggregation window (1 year), \system{}'s noise addition is small (relative to the other queries using a window on the order of hours) and accuracy is high. 

\Para{Case~3: Q7-Q9 (Counting non-private objects, large window)}.
We measure the fraction of trees (non-private objects) that have bloomed in each video.
Executed over an entire network of cameras, such a query could be used
to identify the regions with the best foliage in Spring. 
Relative to Case 1, we achieve high accuracy by using a longer query window of 12 hours (the status of a tree does not change on that time scale), and minimal chunk size (1 frame, no temporal context needed). 

\Para{Case~4: Q10-Q12 (Fine-grained results using aggressive masking)}.
We measure the average amount of time a traffic signal stays red. 
Since this only requires observing the light itself, we can mask \emph{everything else}, resulting in a $\rho$ bound of 0 (no private objects overlap these pixels), enabling high accuracy and fine temporal granularity.

\Para{Case~5: Q13 (Stateful query)}.
We count only the individuals that enter from the south and exit at the north. It requires a larger chunk size (relative to Q1-Q3) to maintain enough state within a single chunk to understand trajectory.

\tightsubsection{Analyzing Sources of Inaccuracy}
\label{eval:main}

\system{} introduces two sources of inaccuracy to a query result: (1) intentional noise to satisfy \pkeprivacy{}, and (2) (unintentional) inaccuracies caused by the impact of splitting and masking videos before executing the video processing.
\Fig{eval:graph2} shows these two sources separately for queries Q1-Q3 (Case 1): the discrepancy between the two curves demonstrates the impact of (2), while the shaded belt shows the relative scale of noise added (1).
In summary, the amount of noise added by \system{} allows the final result to preserve the overall trend of the original. 

\nop{
\begin{figure}
  \centering
  \includegraphics[width=\columnwidth]{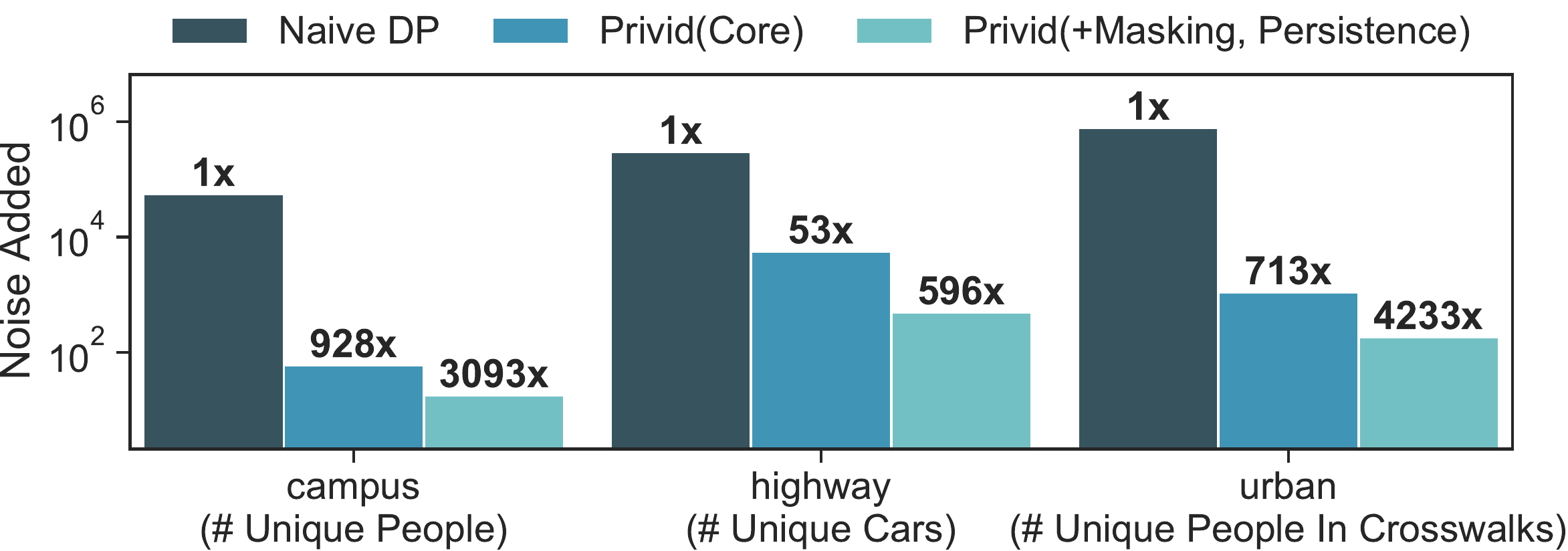}
  \tightcaption{Privid(Core) incorporates the persistence threshold, but not masking or splitting. Privid(+) adds them. The y-axis is the proportion of noise that would be added to the aggregate query result. The number above each bar shows the relative \emph{decrease} in noise compared to Naive.}
  \label{fig:eval:graph1}
\end{figure}
}

\nop{
\Para{Scale of noise.}
To provide context for the amount of noise \system{} adds, we compare to a hypothetical \emph{naive} application of DP, which does not use \pkeprivacy{} and thus sets noise equivalent to the entire output range. 
\Fig{eval:graph1} shows that, \eg{} for \auburn{}, utilizing \system{} results in a three orders of magnitude reduction ($928\times$) in the amount of noise necessary to satisfy DP compared to Naive. Incorporating masking and splitting further improves the reduction $3,093\times$ relative to Naive.
In our example per-hour count query, this level of noise translates to adding or removing roughy 18 people to the total aggregate, compared to 1k people for Naive. 
}

\tightsubsection{Impact of Parameters}
\label{sec:eval:param-sweep}

\system{} provides flexibility for the analyst to balance temporal granularity and accuracy of results (\ie{} a more accurate result over a longer time horizon or vice versa).
To showcase this flexibility, we re-execute the three video/query pairs from Case Study 1 and jointly sweep over a range of chunk size and output range (\Fig{eval:graph3}) and over a range of window sizes (\Fig{eval:graph3b}).

\Fig{eval:graph3} shows that as we increase the chunk size for a given output range, the average error decreases (more context helps raw query accuracy), but the size of the error bar increases (due to additional noise). As one increases the chunk size, it makes the rows of the intermediate table coarser, and thus each row represents a larger fraction of the whole, requiring more noise to cover. 
\extendedonly{Put differently, the effective $\hat{\rho}$ that \system{} enforces becomes larger as the chunk size increases (\Eqtn{rho-hat}).}
For relatively small chunk sizes (less than the persistence), the decrease in error from having more context outweighs the increase in error from slightly larger noise. 

\Fig{eval:graph3b} keeps the chunk size and output range fixed (at the ``X'' values from \Fig{eval:graph3}) and shows that as the window size increases, the amount of noise required to hide an individual (and thus error introduced by \system{}) decreases. The number of chunks an individual could influence remains constant, while the total number of chunks included in the aggregate result grows. 

\tightsection{Ethics}
\label{sec:ethics}

In building \system{}, \emph{we do not advocate for the increase of public video surveillance and analysis}. Instead, we observe that it is already prevalent, and is driven by strong economic and public safety incentives. Consequently, it is undeniable that the analysis of public video will continue, and thus, it is paramount that we provide tools to improve the privacy landscape for such analytics. We seek to encourage video owners that it is indeed possible to have privacy as a first-class citizen, while still enabling useful queries. Further, we anticipate legislation will increasingly restrict video collection and analysis; privacy-preserving systems (like \system{}) will be crucial to enable critical applications while complying with such laws.

\end{sloppypar}
\clearpage

\bibliographystyle{plain}
\interlinepenalty=10000
\bibliography{privid}

\begin{thebibliography}{10}

\bibitem{beijing-cameras}
Absolutely everywhere in beijing is now covered by police video surveillance.
\newblock \url{https://qz.com/518874/}.

\bibitem{are-we-ready-for-ai-powered-security-cameras}
Are we ready for ai-powered security cameras?
\newblock
  \url{https://thenewstack.io/are-we-ready-for-ai-powered-security-cameras/}.

\bibitem{LondonCamera}
British transport police: Cctv.
\newblock
  \url{http://www.btp.police.uk/advice_and_information/safety_on_and_near_the_railway/cctv.aspx}.

\bibitem{ChicagoCamera}
Can 30,000 cameras help solve chicago's crime problem?
\newblock
  \url{https://www.nytimes.com/2018/05/26/us/chicago-police-surveillance.html}.

\bibitem{infowatch-surveillance}
Data generated by new surveillance cameras to increase exponentially in the
  coming years.
\newblock \url{http://www.securityinfowatch.com/news/12160483/}.

\bibitem{coco_scoreboard}
Detection leaderboard.
\newblock \url{https://cocodataset.org/#detection-leaderboard}.

\bibitem{epic_surveillance}
Epic domestic surveillance project.
\newblock \url{https://epic.org/privacy/surveillance/}.

\bibitem{camera-ban-oak}
Oakland bans use of facial recognition.
\newblock
  \url{https://www.sfchronicle.com/bayarea/article/Oakland-bans-use-of-facial-recognition-14101253.php}.

\bibitem{paris-hospital}
Paris hospitals to get 1,500 cctv cameras to combat violence against staff.
\newblock \url{https://bit.ly/2OYiBz2}.

\bibitem{powering-the-edge-with-ai-in-an-iot-world}
Powering the edge with ai in an iot world.
\newblock
  \url{https://www.forbes.com/sites/forbestechcouncil/2020/04/06/powering-the-edge-with-ai-in-an-iot-world/}.

\bibitem{camera-ban-sf}
San francisco is first us city to ban facial recognition.
\newblock \url{https://www.bbc.com/news/technology-48276660}.

\bibitem{smart-mall}
Video analytics applications in retail - beyond security.
\newblock
  \url{https://www.securityinformed.com/insights/co-2603-ga-co-2214-ga-co-1880-ga.16620.html/}.

\bibitem{vision-zero}
The vision zero initiative.
\newblock \url{http://www.visionzeroinitiative.com/}.

\bibitem{aclu-video-surveillance}
What's wrong with public video surveillance?
\newblock
  \url{https://www.aclu.org/other/whats-wrong-public-video-surveillance}, 2002.

\bibitem{camera-spy-abuses}
Abuses of surveillance cameras.
\newblock \url{http://www.notbored.org/camera-abuses.html}, 2010.

\bibitem{google-mission-creep}
Mission creep-y: Google is quietly becoming one of the nation's most powerful
  political forces while expanding its information-collection empire.
\newblock
  \url{https://www.citizen.org/wp-content/uploads/google-political-spending-mission-creepy.pdf},
  2014.

\bibitem{theyarewatching}
Mission creep.
\newblock
  \url{https://www.aclu.org/other/whats-wrong-public-video-surveillance}, 2017.

\bibitem{retail_example}
How retail stores can streamline operations with video content analytics.
\newblock
  \url{https://www.briefcam.com/resources/blog/how-retail-stores-can-streamline-operations-with-video-content-analytics/},
  2020.

\bibitem{streetlights-mission-creep}
The mission creep of smart streetlights.
\newblock
  \url{https://www.voiceofsandiego.org/topics/public-safety/the-mission-creep-of-smart-streetlights/},
  2020.

\bibitem{traffic-analysis}
Video analytics traffic study creates baseline for change.
\newblock
  \url{https://www.govtech.com/analytics/Video-Analytics-Traffic-Study-Creates-Baseline-for-Change.html},
  2020.

\bibitem{cv-frameworks}
What is computer vision? ai for images and video.
\newblock
  \url{https://www.infoworld.com/article/3572553/what-is-computer-vision-ai-for-images-and-video.html},
  2020.

\bibitem{i-pic}
Paarijaat Aditya, Rijurekha Sen, Peter Druschel, Seong Joon~Oh, Rodrigo
  Benenson, Mario Fritz, Bernt Schiele, Bobby Bhattacharjee, and Tong~Tong Wu.
\newblock I-pic: A platform for privacy-compliant image capture.
\newblock In {\em Proceedings of the 14th Annual International Conference on
  Mobile Systems, Applications, and Services}, MobiSys ’16, page 235–248,
  New York, NY, USA, 2016. Association for Computing Machinery.

\bibitem{amazon-rekognition}
Amazon.
\newblock Rekognition.
\newblock \url{https://aws.amazon.com/rekognition/}.

\bibitem{rocket}
Ganesh Ananthanarayanan, Yuanchao Shu, Mustafa Kasap, Avi Kewalramani, Milan
  Gada, and Victor Bahl.
\newblock Live video analytics with microsoft rocket for reducing edge compute
  costs, July 2020.

\bibitem{msft-cv}
Microsoft Azure.
\newblock Computer vision api.
\newblock
  \url{https://azure.microsoft.com/en-us/services/cognitive-services/computer-vision/},
  2021.

\bibitem{azure-face}
Microsoft Azure.
\newblock Face api.
\newblock
  \url{https://azure.microsoft.com/en-us/services/cognitive-services/face/},
  2021.

\bibitem{miris}
Favyen Bastani, Songtao He, Arjun Balasingam, Karthik Gopalakrishnan, Mohammad
  Alizadeh, Hari Balakrishnan, Michael Cafarella, Tim Kraska, and Sam Madden.
\newblock Miris: Fast object track queries in video.
\newblock In {\em Proceedings of the 2020 ACM SIGMOD International Conference
  on Management of Data}, SIGMOD '20, page 1907–1921, New York, NY, USA,
  2020. Association for Computing Machinery.

\bibitem{Bewley2016_sort}
Alex Bewley, Zongyuan Ge, Lionel Ott, Fabio Ramos, and Ben Upcroft.
\newblock Simple online and realtime tracking.
\newblock In {\em 2016 IEEE International Conference on Image Processing
  (ICIP)}, pages 3464--3468, 2016.

\bibitem{pedestrian-detection-iccv15}
Zhaowei Cai, Mohammad Saberian, and Nuno Vasconcelos.
\newblock Learning complexity-aware cascades for deep pedestrian detection.
\newblock In {\em Proceedings of the 2015 IEEE International Conference on
  Computer Vision (ICCV)}, ICCV '15, pages 3361--3369, Washington, DC, USA,
  2015. IEEE Computer Society.

\bibitem{privacycam}
Ankur Chattopadhyay and Terrance~E Boult.
\newblock Privacycam: a privacy preserving camera using uclinux on the blackfin
  dsp.
\newblock In {\em 2007 IEEE Conference on Computer Vision and Pattern
  Recognition}, pages 1--8. IEEE, 2007.

\bibitem{original-dp}
Cynthia Dwork, Frank McSherry, Kobbi Nissim, and Adam Smith.
\newblock Calibrating noise to sensitivity in private data analysis.
\newblock In Shai Halevi and Tal Rabin, editors, {\em Theory of Cryptography},
  volume 3876 of {\em Lecture Notes in Computer Science}, pages 265--284,
  Berlin, Heidelberg, March 2006. Springer.

\bibitem{analyzing-social-distancing}
Isha Ghodgaonkar, Subhankar Chakraborty, Vishnu Banna, Shane Allcroft, Mohammed
  Metwaly, Fischer Bordwell, Kohsuke Kimura, Xinxin Zhao, Abhinav Goel, Caleb
  Tung, et~al.
\newblock Analyzing worldwide social distancing through large-scale computer
  vision.
\newblock {\em arXiv preprint arXiv:2008.12363}, 2020.

\bibitem{google-cloud-vision}
Google.
\newblock Cloud vision api.
\newblock \url{https://cloud.google.com/vision}, 2021.

\bibitem{focus}
Kevin Hsieh, Ganesh Ananthanarayanan, Peter Bodik, Shivaram Venkataraman,
  Paramvir Bahl, Matthai Philipose, Phillip~B Gibbons, and Onur Mutlu.
\newblock Focus: Querying large video datasets with low latency and low cost.
\newblock In {\em 13th USENIX Symposium on Operating Systems Design and
  Implementation (OSDI 18)}, pages 269--286, 2018.

\bibitem{ibm-maximo}
IBM.
\newblock Maximo remote monitoring.
\newblock \url{https://www.ibm.com/products/maximo/remote-monitoring}, 2021.

\bibitem{rexcam-hotmobile}
Samvit Jain, Ganesh Ananthanarayanan, Junchen Jiang, Yuanchao Shu, and
  Joseph~E. Gonzalez.
\newblock {Scaling Video Analytics Systems to Large Camera Deployments}.
\newblock In {\em ACM HotMobile}, 2019.

\bibitem{spatula}
Samvit Jain, Xun Zhang, Yuhao Zhou, Ganesh Ananthanarayanan, Junchen Jiang,
  Yuanchao Shu, Victor Bahl, and Joseph Gonzalez.
\newblock Spatula: Efficient cross-camera video analytics on large camera
  networks.
\newblock In {\em ACM/IEEE Symposium on Edge Computing (SEC 2020)}, November
  2020.

\bibitem{chameleon}
Junchen Jiang, Ganesh Ananthanarayanan, Peter Bodik, Siddhartha Sen, and Ion
  Stoica.
\newblock Chameleon: scalable adaptation of video analytics.
\newblock In {\em Proceedings of the 2018 Conference of the ACM Special
  Interest Group on Data Communication}, pages 253--266. ACM, 2018.

\bibitem{flex}
Noah Johnson, Joseph~P Near, and Dawn Song.
\newblock Towards practical differential privacy for sql queries.
\newblock {\em Proceedings of the VLDB Endowment}, 11(5):526--539, 2018.

\bibitem{kairouz}
Peter Kairouz, Sewoong Oh, and Pramod Viswanath.
\newblock The composition theorem for differential privacy.
\newblock {\em IEEE Transactions on Information Theory}, 63(6):4037--4049,
  2017.

\bibitem{blazeit}
Daniel Kang, Peter Bailis, and Matei Zaharia.
\newblock Blazeit: optimizing declarative aggregation and limit queries for
  neural network-based video analytics.
\newblock {\em Proceedings of the VLDB Endowment}, 13(4):533--546, 2019.

\bibitem{noscope}
Daniel Kang, John Emmons, Firas Abuzaid, Peter Bailis, and Matei Zaharia.
\newblock Noscope: optimizing neural network queries over video at scale.
\newblock {\em Proceedings of the VLDB Endowment}, 10(11):1586--1597, 2017.

\bibitem{tasti}
Daniel Kang, John Guibas, Peter Bailis, Tatsunori Hashimoto, and Matei Zaharia.
\newblock Task-agnostic indexes for deep learning-based queries over
  unstructured data.
\newblock {\em arXiv preprint arXiv:2009.04540}, 2020.

\bibitem{privatesql}
Ios Kotsogiannis, Yuchao Tao, Xi~He, Maryam Fanaeepour, Ashwin Machanavajjhala,
  Michael Hay, and Gerome Miklau.
\newblock Privatesql: A differentially private sql query engine.
\newblock {\em Proc. VLDB Endow.}, 12(11):1371–1384, July 2019.

\bibitem{imagenet-classification-cacm17}
Alex Krizhevsky, Ilya Sutskever, and Geoffrey~E. Hinton.
\newblock Imagenet classification with deep convolutional neural networks.
\newblock {\em Commun. ACM}, 60(6):84--90, May 2017.

\bibitem{cnn-face-cvpr15}
H.~{Li}, Z.~{Lin}, X.~{Shen}, J.~{Brandt}, and G.~{Hua}.
\newblock A convolutional neural network cascade for face detection.
\newblock In {\em 2015 IEEE Conference on Computer Vision and Pattern
  Recognition (CVPR)}, pages 5325--5334, June 2015.

\bibitem{reducto}
Yuanqi Li, Arthi Padmanabhan, Pengzhan Zhao, Yufei Wang, Guoqing~Harry Xu, and
  Ravi Netravali.
\newblock {Reducto: On-Camera Filtering for Resource-Efficient Real-Time Video
  Analytics}.
\newblock SIGCOMM '20, page 359–376, New York, NY, USA, 2020. Association for
  Computing Machinery.

\bibitem{pyramid-network-cvpr17}
T.~{Lin}, P.~{Doll{\'a}r}, R.~{Girshick}, K.~{He}, B.~{Hariharan}, and
  S.~{Belongie}.
\newblock Feature pyramid networks for object detection.
\newblock In {\em 2017 IEEE Conference on Computer Vision and Pattern
  Recognition (CVPR)}, pages 936--944, July 2017.

\bibitem{porto}
Luis Moreira-Matias, Joao Gama, Michel Ferreira, Joao Mendes-Moreira, and Luis
  Damas.
\newblock Predicting taxi--passenger demand using streaming data.
\newblock {\em IEEE Transactions on Intelligent Transportation Systems},
  14(3):1393--1402, 2013.

\bibitem{denaturing-survey}
Jos{\'e}~Ram{\'o}n Padilla-L{\'o}pez, Alexandros~Andre Chaaraoui, and Francisco
  Fl{\'o}rez-Revuelta.
\newblock Visual privacy protection methods: A survey.
\newblock {\em Expert Systems with Applications}, 42(9):4177--4195, 2015.

\bibitem{visor}
Rishabh Poddar, Ganesh Ananthanarayanan, Srinath Setty, Stavros Volos, and
  Raluca~Ada Popa.
\newblock Visor: Privacy-preserving video analytics as a cloud service.
\newblock In {\em 29th $\{$USENIX$\}$ Security Symposium ($\{$USENIX$\}$
  Security 20)}, pages 1039--1056, 2020.

\bibitem{faster_rcnn}
Shaoqing Ren, Kaiming He, Ross~B. Girshick, and Jian Sun.
\newblock Faster {R-CNN:} towards real-time object detection with region
  proposal networks.
\newblock {\em CoRR}, abs/1506.01497, 2015.

\bibitem{stanley2019dawn}
J.~Stanley and American Civil~Liberties Union.
\newblock {\em The Dawn of Robot Surveillance: AI, Video Analytics, and
  Privacy}.
\newblock American Civil Liberties Union, 2019.

\bibitem{facial-point-cvpr13}
Yi~Sun, Xiaogang Wang, and Xiaoou Tang.
\newblock Deep convolutional network cascade for facial point detection.
\newblock In {\em Proceedings of the 2013 IEEE Conference on Computer Vision
  and Pattern Recognition}, CVPR '13, pages 3476--3483, Washington, DC, USA,
  2013. IEEE Computer Society.

\bibitem{verro}
Han Wang, Yuan Hong, Yu~Kong, and Jaideep Vaidya.
\newblock Publishing video data with indistinguishable objects.
\newblock {\em Advances in database technology : proceedings. International
  Conference on Extending Database Technology}, 2020:323 -- 334, 2020.

\bibitem{videodp}
Han Wang, Shangyu Xie, and Yuan Hong.
\newblock Videodp: A universal platform for video analytics with differential
  privacy.
\newblock {\em arXiv preprint arXiv:1909.08729}, 2019.

\bibitem{openface}
Junjue Wang, Brandon Amos, Anupam Das, Padmanabhan Pillai, Norman Sadeh, and
  Mahadev Satyanarayanan.
\newblock A scalable and privacy-aware iot service for live video analytics.
\newblock In {\em Proceedings of the 8th ACM on Multimedia Systems Conference},
  pages 38--49. ACM, 2017.

\bibitem{google-dp}
Royce~J Wilson, Celia~Yuxin Zhang, William Lam, Damien Desfontaines, Daniel
  Simmons-Marengo, and Bryant Gipson.
\newblock Differentially private sql with bounded user contribution.
\newblock {\em Proceedings on privacy enhancing technologies},
  2020(2):230--250, 2020.

\bibitem{deepsort}
Nicolai Wojke, Alex Bewley, and Dietrich Paulus.
\newblock Simple online and realtime tracking with a deep association metric.
\newblock In {\em 2017 IEEE International Conference on Image Processing
  (ICIP)}, pages 3645--3649. IEEE, 2017.

\bibitem{pecam}
Hao Wu, Xuejin Tian, Minghao Li, Yunxin Liu, Ganesh Ananthanarayanan, Fengyuan
  Xu, and Sheng Zhong.
\newblock Pecam: Privacy-enhanced video streaming and analytics via
  securely-reversible transformation.
\newblock In {\em ACM MobiCom}, October 2021.

\bibitem{detectron_v2}
Yuxin Wu, Alexander Kirillov, Francisco Massa, Wan-Yen Lo, and Ross Girshick.
\newblock Detectron2.
\newblock \url{https://github.com/facebookresearch/detectron2}, 2019.

\bibitem{prisurv}
Xiaoyi Yu, Kenta Chinomi, Takashi Koshimizu, Naoko Nitta, Yoshimichi Ito, and
  Noboru Babaguchi.
\newblock Privacy protecting visual processing for secure video surveillance.
\newblock In {\em 2008 15th IEEE International Conference on Image Processing},
  pages 1672--1675. IEEE, 2008.

\bibitem{videostorm}
Haoyu Zhang, Ganesh Ananthanarayanan, Peter Bodik, Matthai Philipose, Paramvir
  Bahl, and Michael~J Freedman.
\newblock Live video analytics at scale with approximation and delay-tolerance.
\newblock In {\em NSDI}, volume~9, page~1, 2017.

\bibitem{zhu2017flow}
Xizhou Zhu, Yujie Wang, Jifeng Dai, Lu~Yuan, and Yichen Wei.
\newblock Flow-guided feature aggregation for video object detection.
\newblock In {\em Proceedings of the IEEE International Conference on Computer
  Vision}, pages 408--417, 2017.

\end{thebibliography}
\clearpage

\newpage

\begin{appendix}

\section{Conservatively Estimating Durations}\label{app:cv-params}

\begin{table*}
    \centering
\begin{tabular}{|l|l|l|l|l|}
\hline
\textbf{video} & \textbf{cos}                     & \textbf{iou}                     & \textbf{age}                         & \textbf{n\_init}   \\ \hline
campus (0.8)                 & 0.1, 0.3, \textbf{0.5}, 0.7, 0.9 & 0.1, 0.3, 0.5, \textbf{0.7}, 0.9 & 16,  32, 48, 64, 80, \textbf{96}, 112 & 2, 3, 5, 7, \textbf{9} \\ \hline
urban (0.6)                  & \textbf{0.1}, 0.3, 0.5, 0.7, 0.9 & 0.1, 0.3, \textbf{0.5}, 0.7, 0.9 & 8, 16, 32, 48, 64, 80, \textbf{96}   & 2, 3, \textbf{5}, 7, 9 \\ \hline
\end{tabular}
    \caption{Set of hyperparameters used for tuning DeepSORT for the \auburn{} and \shibuya{} videos. The set of parameters that we ultimately used for our experiments are bolded.}
\label{tab:deepsort-params}
\end{table*}

\begin{table*}
    \centering
\begin{tabular}{|l|l|l|l|}
\hline
\textbf{video} & \textbf{max\_age}      & \textbf{min\_hits}  & \textbf{iou\_dist}          \\ \hline
highway (0.2)                & 240, 480, \textbf{720} & 3, 5, 7, \textbf{9} & \textbf{0.1}, 0.3, 0.5, 0.7 \\ \hline
\end{tabular}
    \caption{Set of hyperparameters used for tuning SORT for the \hampton{} video. The set of parameters that we ultimately used for our experiments are bolded.}
\label{tab:sort-params}
\end{table*}

Estimating duration values for a given scene requires the ability to track individuals in that scene. Unfortunately, even state-of-the-art vision techniques for object tracking are riddled with inaccuracies that stem from occlusion (\ie{} line of sight to an object is blocked), illumination, and poor video quality; these challenges are exacerbated in low-quality public surveillance videos. Manual annotation of individuals in video can overcome these challenges but is far from scalable and is difficult to use for real-time video analysis.

We observe that, even though the aforementioned challenges preclude off-the-shelf algorithms from perfectly tracking every individual, their hyperparameters can be trained in a way that generates a reasonably accurate distribution of duration values, which is sufficient for \system{} to provide meaningful privacy guarantees.

For each of the three video in our dataset, we first ran object detection using Facebook's Detectron2 \cite{detectron_v2} library with the included Faster-RCNN model \cite{faster_rcnn}. Using these object detection results, we then manually annotated a subset 
of video for each camera, producing a ground truth dataset of duration values. Annotation for a video involved recording the exact time each unique individal entered and exited the scene at the second granularity. Individuals may reappear and thus have multiple enter and exit times. 

Using our ground truth dataset, we then tuned the hyperparameters of a state-of-the-art tracking algorithm called DeepSORT \cite{deepsort} for each camera's video. Our goal was to find the configuration of parameters that produced the distribution of duration values which most closely matched that of the annotated ground truth data. To do this, we ran DeepSORT with all possible combinations of the hyperparameters listed in Table~\ref{tab:deepsort-params}. For each configuration, we computed the distribution of duration values, and compared it to our ground truth distribution. 

In \hampton{}, we consider cars as the private object rather than people because no
people are visible in the video, but a car's license plate, or their combination
of make, model and color may be enough to identify an individual. 
As DeepSORT is specific to tracking people, we used SORT \cite{Bewley2016_sort} instead. Table~\ref{tab:sort-params} lists the set of hyperparameters we considered and chose for tuning SORT. In practice, if a video contains both people and cars, the persistence distribution should account for both.

\section{Isolated Execution Requirements}
\label{app:isolated-execution}

In order for \system{}'s privacy guarantees to hold, it must run an independent instantiation of the analyst's \texttt{PROCESS} executable for each chunk inside an isolated execution environment. Abstractly, this environment must ensure:
    \begin{enumerate}
        \item The query output for a chunk is based solely on information in that chunk, and not information in any other chunk. 
        \item The query output is the only information the analyst can observe (\ie{} there are no side channels that can leak information beyond what \system{} is expecting)
    \end{enumerate}

In practice, this requires accounting for a number of possible subtle mechanisms that could be used to communicate between executions of a sandbox or outside of the sandbox entirely. The following is a list of requirements on the execution process to ensure the above:
\begin{itemize}
    \item The process must not be able to read or write from the network or any IPC mechanisms.
    \item The process must not be able to access or create files that are visible to another process. 
    \item Access to a PRNG (\eg{} \texttt{/dev/urandom}) must be cryptographically secure.
        Otherwise, if the sequence of bits were predictable, a writer execution could read from
        the generator until the desired bit is ready and then stop. If the reader can read
        after the writer has finished, it will see the writer's bit. 
    \item If the process' resource usage is monitored (\eg{} for billing purposes), it must not
        be made available to the analyst. It is reasonable to assume that, since the video
        owner is running the computation on behalf of the analyst, they may wish for the
        analyst to pay or the computation. In that case resource usage and cost must be
        determined entirely a priori and cannot depend on the actual execution itself,
        otherwise the precise resource usage could be a side channel. 
    \item The process must not be able to vary its own execution time in a way that is visible
        to the analyst. For example, a malicious query could exit immediately if
        $x$ is not present, but spin the processor for a long time if they are. This would
        cause a noticeable increase in total execution time compared to a segment of video where
        $x$ was not present, leaking $x$'s presence.
\end{itemize}

Formally, this environment can be modeled as a turing machine with the following properties (where $R$, $T$, \texttt{schema}, and \texttt{default} are specified a priori as part of the query $\query$):
\begin{itemize}
    \item Takes as \textbf{input}: a set of frames representing a chunk, timestamp of the first frame, frame rate (in fps), camera ID, and (optionally) any additional meta-data the video owner wishes to provide that is not dependent on any private information, such as the amount of daylight at that time.
    \item Has access to a \textbf{random tape} that is uncorrelated with any of the other chunks
    \item Produces as \textbf{output} at most $R$ rows, each with columns specified by \texttt{schema}.
    \item Executes for exactly $T$ seconds. If it finishes early, it must wait until $T$ seconds have elapsed. If it does not finish in time or crashes, it produces a \texttt{default} value.
\end{itemize}

\section{Degradation of Privacy}
\label{app:graceful-degradation}

\begin{figure}
  \centering
  \includegraphics[width=\columnwidth]{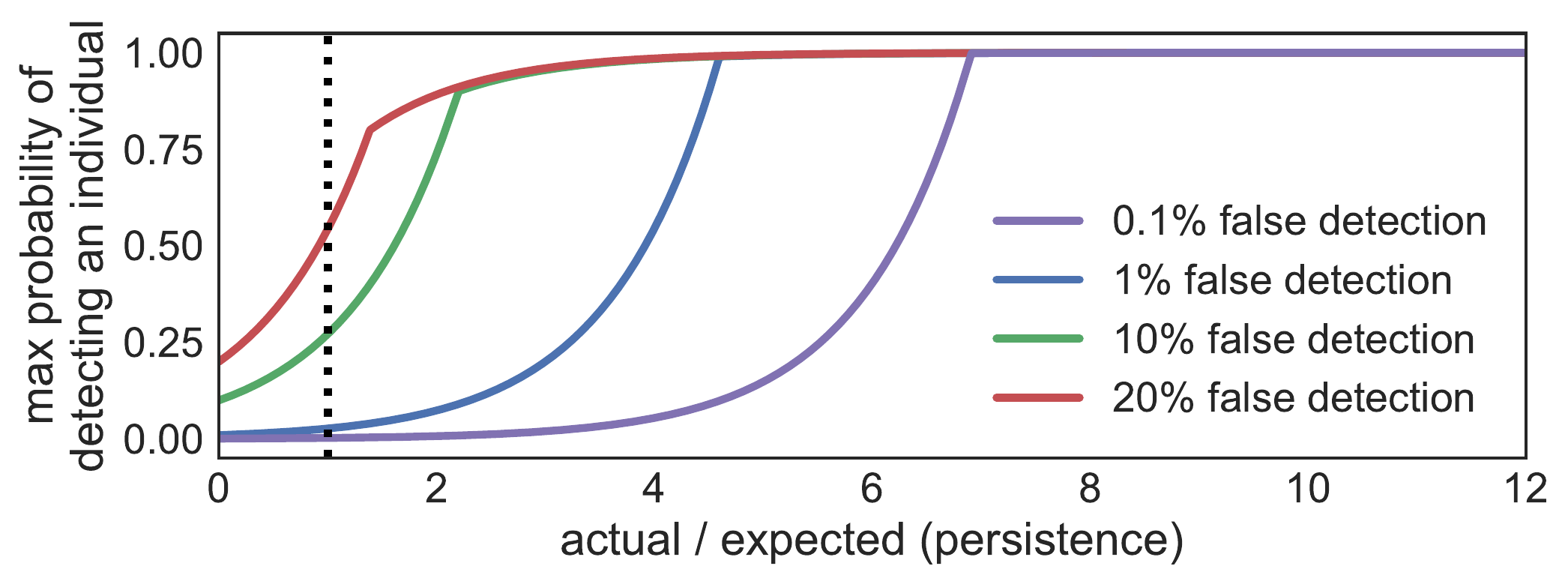}
  \vspace{-3pt}
  \tightcaption{Plot of~\Eqtn{p-detect} for a few different levels of $\alpha$. When an event (\eg{} the presence of a person) exceeds the $\pk{}$ bound protected by \system{}, their presence is not immediately revealed. Rather, as it exceeds the bound further, it becomes more likely an adversary could detect the event; here, we characterize that relationship. The $x$-axis plots an individual's (actual) $\rho$ \emph{relative} to the bound guaranteed (expected $\rho$); \ie{} at $x=2$, the individual is visible for $2\rho$. The $y$-axis plots the maximum probability that an adversary with a given confidence level could detect whether or not the event occurred.}
  \label{fig:privacy-degrade}
\end{figure}

A nice property of differentially-private algorithms is that privacy ``degrades gracefully'':
coming ``close'' to satisfying the definition of privacy, but not all the way, still provides a strong
level of privacy due to the randomness of the noise component. With \system{}, the $\pk$ bound
is the point at which the adversary could begin to do better than random guessing to determine
the presence of a \pkbounded{} event. 

We can formalize this using the framework of binary hypothesis testing. 
Consider an adversary who wishes to determine whether or not $x$ appeared in a given video $V$.
They submit a query to the system, and observe only the final result,
$A$, which \system{} computed as $A = r + \eta$. 
Based on this value, they must distinguish between the two hypotheses:

\begin{align*}
    \mathcal{H}_0:~&\text{$x$ does not appear in $V$} \\
    \mathcal{H}_1:~&\text{$x$ appears in $V$}
\end{align*}

We write the false positive $P_{FP}$ and false negative $P_{FN}$ probabilities as:
\begin{align*}
    P_{FP} &= \mathbb{P}(x \in V | \mathcal{H}_0) \\ 
    P_{FN} &= \mathbb{P}(x \notin V | \mathcal{H}_1)
\end{align*}

From Kairouz~\cite[Theorem~2.1]{kairouz}, if an algorithm guarantees $\epsilon$-differential privacy ($\delta=0$), then these probabilities are related as follows:

\begin{align}
\label{eq:fp}
    P_{FP} + e^{\epsilon}P_{FN} &\ge 1 \\
\label{eq:fn}
    P_{FN} + e^{\epsilon}P_{FP} &\ge 1
\end{align}

Suppose the adversary is willing to accept a false positive threshold of 
$P_{FP} \le \alpha$. 
In ther words, they will only accept $\mathcal{H}_1$ ($x$ is present) 
if there is less than $\alpha$ probability that $x$ is not actually present.

Rearranging equations \ref{eq:fp} and \ref{eq:fn} in terms of the probability of 
correctly detetecting $x$ is present ($1 - P_{FN}$), we have:
\begin{align*}
    1 - P_{FN} &\le& e^{\epsilon}P_{FP}                       &\le& e^{\epsilon}\alpha \\
    1 - P_{FN} &\le& e^{-\epsilon}(P_{FP} - (1-e^{\epsilon})) &\le& e^{-\epsilon}(\alpha - (1-e^{\epsilon}))
\end{align*}

The probability that the adversary correctly decides $x$ is present is then at most the minimum
of these:
\begin{equation}
    \label{eq:p-detect}
    \mathbb{P}(x \in V | \mathcal{H}_1) \le \min\{e^{\epsilon}\alpha,
    e^{-\epsilon}(\alpha-(1-e^{\epsilon}))\}
\end{equation}

In \Fig{privacy-degrade}, we visualize~\ref{eq:p-detect} as a function of an individual's persistence
past $\rho$, for 4 different adversarial confidence levels ($\alpha$=0.1\%,1\%,10\%,20\%).

\extendedonly{
\section{Relative Privacy Guarantee}
\label{app:pke-proportional}


Consequently, since the ceiling operation ``round''s the duration given by a policy up to the nearest chunk, \system{} actually protects $(\hat{\rho},K)$-bounded events with $\hat{\rho} = \maxchunks{\rho} \cdot c_t$ seconds, which is always greater than $\rho$.
}

\section{Query Grammar}
\label{app:grammar}

\begin{figure}
    \centering
    \small
    \begin{lstlisting}
query := split_stmt | process_stmt | select_stmt 
split_stmt := SPLIT camera_id
    BEGIN timestamp 
    END timestamp
    BY TIME chunk_sec STRIDE stride_sec
    [BY REGION ...] // optional, see Section 7
    [WITH MASK ...] // optional, see Section 7
    INTO chunk_set_id;
process_stmt := PROCESS chunk_set_id
    USING binary_name
    TIMEOUT timeout_sec
    PRODUING maxrows 
    WITH SCHEMA chunk_output_schema
    INTO table_id;
select_stmt := outer_select FROM inner_select 
    [GROUP BY col_list WITH KEYS ...]
outer_select := SELECT agg_fun(col_name)
inner_select := table_id | process_stmt
    | SELECT expr_list FROM inner_select 
        [WHERE condition] [LIMIT rows]
    | inner_select GROUP BY col_list [WITH KEYS ...]
    | inner_select JOIN inner_select ON col_list 
chunk_output_schema := /* list of */ col_name:dtype=default
agg_fun := SUM | COUNT | AVG | ...
expr := col_name | expr + expr | expr * expr | ...
dtype := STRING | NUMBER
     \end{lstlisting}
    \caption{\System Query Grammar. Terms in capital letters are query language keywords. Keywords in square brackets are optional. The term {\tt col\_name} stands for the name of an analyst-provided column.}
    \label{fig:query-grammar}
\end{figure}

A \texttt{split\_stmt} converts a segment of video data from a single camera into a named set of chunks by specifying the following: 
\begin{itemize}[leftmargin=*]
    \item The \texttt{BEGIN} and \texttt{END} timestamps describe the bounds of time the analyst is interested in. Tables are evaluated lazily only once they are needed for an aggregation so the analyst can choose large time bounds (\eg{}, an entire year) but narrow to specific times (\eg{}, 1 hour per weekday) using the aggregation statement. These times may be in the past or future (\ie{} for streaming queries). Any values that only depend upon past timestamps will be processed and released as soon as possible (limited only by the processing time requirements described in~\App{isolated-execution}). Any values that depend upon future timestamps will be released as soon as possible (given the timeout) after all of the timestamps needed have elapsed. 
    \item \texttt{BY TIME} describes the duration of each chunk, and the {\tt STRIDE} between chunks. Both values must correspond to an integer number of frames (\eg{} at a frame rate of 30 fps, 0.5 seconds is permitted because it corresponds to 15 frames, but 0.25 seconds is not permitted because it corresponds to 7.5 frames). The chunk duration must be positive, but the stride may be 0 or even negative (for overlapping chunks). 
    \item \texttt{BY REGION} describes the scheme used to further split each chunk spatially. The video owner defines (and publicly releases) a set of schemes. 
    \item \texttt{WITH MASK} specifies the id of a video-owner-provided mask. A mask specifies a set of pixels to remove from the video (\ie{} replace with black pixels). This mask is applied to the video before splitting, and thus before the analyst's executable is able to view the video. 

A \texttt{process\_stmt} uses the analyst-provided executable to convert a set of chunks (created by a \texttt{split\_stmt}) into an intermediate table by specifying the following:
    \item \texttt{USING} provides the path of the analyst-provided executable that should be used to process each chunk of this camera's video data. Analysts may supply any number of executables and use different executables for different cameras. These executables take as input a list of (contiguous) frames from the video, and output rows of a table (whose schema is defined by the \texttt{PRODUCE} directive). Each chunk is processed by an independent instantiation of the executable in a confined execution environment.
    \item \texttt{TIMEOUT} specifies the maximum amount of time that can be used to process each chunk. If execution exceeds this time for any chunk, it is immediately terminated and a row is output with the \texttt{default} values for each column as specified in the \texttt{user\_schema}. 
    The existence of the \texttt{TIMEOUT} clause is 
    crucial for preventing side-channel information leakage via the execution time. 
\item \texttt{PRODUCING \emph{maxrows} WITH SCHEMA \emph{schema}} specifies the schema of columns in the table and the maximum number of rows each chunk will output. For each column, the schema specifies a name (for reference in aggregations), a data type (either \texttt{STRING} or \texttt{NUMBER}, used to determine the types of aggregations permitted over the column), and a default value (to be output if the processing for that chunk crashes or exceeds \texttt{TIMEOUT}). \system{} does not place any trust in the executable or make any assumptions about the content of the output; it truncates the output as necessary to ensure that it adheres to the schema.

In addition to the user-specified columns, \system{} also adds a \texttt{chunk} column to every table which contains the timestamp of the first frame of the chunk. This can be used to narrow time ranges (\eg{} only 12pm-2pm each weekday), aggregate over different amounts of time (\eg{} group results per hour), or match times across cameras or days. 
\end{itemize}

\Para{Table Selection and Aggregation.} A selection-aggregation statement {\tt select\_stmt} computes aggregate statistics from intermediate tables (produced by {\tt process\_stmt}s) using familiar SQL syntax, with some important restrictions to properly control sensitive data leakage: 
\begin{itemize}[leftmargin=*]
    \item The outer-most select (\texttt{outer\_select}) must be an aggregation. Each aggregation must be over a single column (with the exception of \texttt{COUNT(*)}) and is treated as an independent result $r_i$. \system{} uses the Laplace mechanism to add an independent sample of noise to each $r_i$ before releasing it to the analyst, and subtracts from the privacy budget for each $r_i$ as well. The select can optionally group results using a \texttt{GROUP BY}, but only if it explicitly provides the keys (using \texttt{WITH KEYS [...]}, so that they are not dependent on the data) or groups over the \texttt{chunk} column (which \system{} created and therefore can trust). Figure~\ref{fig:sensitivity-table} lists the supported aggregation functions and some restrictions.
    \item An {\tt inner\_select} statement is nested inside an {\tt outer\_select} statement and can be nested inside other {\tt inner\_select} statements.  An {\tt inner\_select} may transform the original table into a new one, combine multiple tables, and select and project rows and columns. 
    \item Some aggregation functions require the range of a column or the number of rows to be constrained (Figure~\ref{fig:sensitivity-table}). When these cannot be inferred, they must be explicitly provided as part of the select via the \texttt{range(col, low, high)} function or the \texttt{LIMIT rows} directive, respectively. 
    \item \system{} includes helper functions, such as \texttt{hour(chunk)} or \texttt{day(chunk)}, which convert the chunk timestamp into the corresponding hour or day. We note their existence simply because they make queries much easier to read.
\end{itemize}

\extendedonly{
\section{Executable for Listing~\ref{lst:car-query-example}}
\label{app:example-executable}

\begin{lstlisting}[language=python, breaklines=true, basicstyle=\ttfamily\scriptsize,  
keywordstyle=\color{green!50!black}\bfseries, caption={Executable \texttt{model.py} used by the \texttt{PROCESS} command of Listing~\ref{lst:car-query-example}, \SecNS{system:query}}, captionpos=b, label={lst:car-query-executable}]
import detectron
import deepsort

tracker = deepsort.Tracker()

for frame in sys.stdin:
    objects = detectron.detect(frame)
    for car in filter(objects, label="car"):
        plate = openalpr.process(car)
        car.plate = plate
        color = compute_obj_color(car)
        car.color = color
    tracker.add(objects)

for car in tracker:
    print(car.plate, car.color, car.speed)
\end{lstlisting}

\system{} expects \texttt{PROCESS} executables to be standalone executables. They must entirely embed any libraries or ML models they use. We can use \texttt{cython3} to achieve this:

\begin{lstlisting}[language=bash, breaklines=true, captionpos=b, label={lst:executable-compile}, caption={Steps necessary to compile python script to standalone executable, which is attached to written \system{} query.}]
> cython3 --embed -o model.c model.py
> gcc -I /usr/include/python3 model.c -lpython3 -o model
\end{lstlisting}
}

\extendedonly{
\section{Full Mechanism Algorithm}
\label{app:full-algorithm}

\fc{Update varaible names, match with Alg 1}
In~\Sec{algorithm}, we presented simplified pseudocode for the \system{} mechanism, assuming a single \texttt{CREATE} statement, a single \texttt{SELECT} statement, and a single video. In~\Fig{full-algorithm}, we present pseudocode for the full mechanism, handling multiple creation and aggregration statements over multiple videos. The generalization from the initial algorithm is straightforward, but we include it here for completeness.

\begin{algorithm}
\SetKwInOut{Input}{Input}
\SetKwInOut{Output}{Output}

\Input{\system{} query $Q = \{[F...], [S...], c, \epsilon_Q\}$, videos $V$, policy $(\rho, K, \epsilon)$}
\Output{Query answer(s) $A$}

\tcp{Ensure sufficient budget for entire interval and $\rho$ margin for \emph{all} aggregations}
\ForEach{$s \in S$}{
	\ForEach{$f \in s.V[s.I\pm\rho]$}{
		\If{$f.budget < \epsilon_Q$}{
			\Return \texttt{DENY}
		}
	}
}
\tcp{There is enough budget for all aggregations, so query is permitted, decrement budget for entire interval for \emph{all} aggregations}
\ForEach{$s \in S$}{
	\ForEach{$f \in s.V[s.I]$}{
		$f.budget$ \texttt{-=} $\epsilon_Q$
	}
}

\tcp{Create intermediate tables $T$}
\ForEach{$c \in C$}{
	\texttt{itable} $\gets$ Table(c.output\_schema)

	\texttt{chunks} $\gets$ Split $V[I]$ into sequential segments each of length \texttt{c.chunk\_sec} with stride \texttt{c.chunk\_stride\_sec}

	\ForEach{$chunk \in chunks$}{
		\texttt{rows} $\gets$ \texttt{F[c.process\_using]}(chunk) \tcp{Executed in confined environment}
		\texttt{itable.append}(rows)
	}
	\texttt{tables[c.table\_name] $\gets$ itable}
}
\tcp{Compute output and add noise}

\ForEach{$(i,s) \in S$}{
	$r \gets$ execute SQL query $s$ over table(s) in $T$, includes joins etc.
		
	$\Delta_{(\rho,K)} \gets $ compute recursively over $s$, using rules in \Tab{sensitivity-rules} \Eqtn{base-table-sensitivity}
		
	$\eta \gets Laplace(\mu=0, b=\frac{\Delta}{\epsilon_Q})$

	$A_i \gets r + \eta$
}

\caption{\system{} Mechanism (full)}
\label{alg:expanded}
\end{algorithm}
}
\section{\system{} Sensitivity Calculation}

\subsection{Propagation Rules}
\label{app:sensitivity}

Figure~\ref{fig:sensitivity-table} provides the full set of rules \system{} uses to compute the sensitivity of a query. 

\newcommand{\STAB}[1]{\begin{tabular}{@{}c@{}}#1\end{tabular}}

\def\ojoin{\setbox0=\hbox{$\Join$}%
\rule[0.0575ex]{.2em}{.4pt}\llap{\rule[1.0775ex]{.2em}{.4pt}}}
\def\leftouterjoin{\mathbin{\ojoin\mkern-5.8mu\Join}}
\def\rightouterjoin{\mathbin{\Join\mkern-5.8mu\ojoin}}
\def\fullouterjoin{\mathbin{\ojoin\mkern-6.8mu\Join\mkern-6.8mu\ojoin}}

\begin{figure*}
\setlength{\extrarowheight}{4pt}
\begin{center}
    \begin{minipage}[t]{0.01\textwidth}
\strut\end{minipage}%
\hfill\allowbreak%
\begin{minipage}[t]{0.40\textwidth}
\tiny
\begin{tabular}{l|p{0.1\textwidth}|p{0.75\textwidth}|}
\cline{2-3}
\multirow[c]{8}{*}{\STAB{\rotatebox[origin=c]{90}{\textsc{Notation}}}}
    & $\policy$ & Privacy policy for each camera: $\{(\rho, K)_c~~\forall~~c~~\in~~\text{cameras}\}$ \\
    & $\Delta_{\policy}(R)$ & Maximum number of rows in relation $R$ that could differ by the addition or removal of any \pkbounded{} event. \\
    & $\rangecons{}(R, a)$ & Range constraint: range of attribute $a$ in $R$ \\
    & $\sizecons{}(R)$ & Size constraint: upper bound on total number of rows in $R$ \\
    & $\varnothing$ & Indicates that a relational operator leaves a constraint unbound. If this constraint is required
    for the aggregation, it must be bound by a predecessor. If it is not required, it can be left unbound. \\
\cline{2-3}
\end{tabular}\label{tab:notation}
\strut\end{minipage}%
\hfill\allowbreak%
\begin{minipage}[t]{0.50\textwidth}
\tiny
\setlength{\extrarowheight}{4pt}
\begin{tabular}{l|l|l|l|l|}
\cline{2-5}

\multirow[c]{7}{*}{\STAB{\rotatebox[origin=c]{90}{\textsc{Aggregation Functions}}}} &
\textbf{Function} &
\textbf{Definition} &
\textbf{Constraints} &
\textbf{Sensitivity ($\mathbf{\Delta(Q)}$)} \\

\cline{2-5}

& Count
	& $Q := \Pi_{\text{count}(*)}(R)$
	& $\Delta$
	& $1 \cdot \D{R}$ \\

\cline{2-5}

& Sum
	& $Q := \Pi_{\text{sum}(a)}(R)$
    & $\Delta, \tilde{C_r}$
	& $\D{R} \cdot \Cr{R}{a}$ \\

\cline{2-5}

& Average
	& $Q := \Pi_{\text{avg}(a)}(R)$
    & $\Delta, \tilde{C_r}, \tilde{C_s}$
	& $\frac{\D{R} \cdot \Cr{R}{a}}{\Cs{R}}$ \\

\cline{2-5}

& Variance
	& $Q := \Pi_{\text{var}(a)}(R)$
    & $\Delta, \tilde{C_r}, \tilde{C_s}$
	& $\frac{(\D{R} \cdot \Cr{R}{a})^2}{\Cs{R}}$ \\

\cline{2-5}

& Argmax
	& $Q := \Pi_{\text{argmax}(a)}(R)$
	& $\Delta, a \in K$
	& $\max_{k \in K}{\D{\sigma_{a=k}(R)}}$ \\

\cline{2-5}

\end{tabular}\label{tab:aggfs}
\end{minipage}
\end{center}

\begin{center}
\tiny
    \begin{tabular}{l|c|l|l|l|l|l|}

\cline{2-7}

\multirow[c]{15}{*}{\STAB{\rotatebox[origin=c]{90}{\textsc{Relational Operators}}}} &

\textbf{Operator} &
\textbf{Type} &
\textbf{Definition} &
$\mathbf{\Dp{R'}}$ &
$\mathbf{\Cr{R'}{a_i}}$ &
$\mathbf{\Cs{R'}}$  \\

\cline{2-7}

  & \multirow{2}{*}{\shortstack{Selection\\($\sigma$)}}

	& \makecell[l]{Standard selection: rows from $R$ that match WHERE condition}
	& $R' := \sigma_{\textsc{where}(\ldots)}(R)$ 
	& $\Dp{R}$
	& $\Cr{R}{a_i}$
	& $\Cs{R}$ \\

	\hhline{~|~|---|--}
	&
	& \makecell[l]{Limit: first $x$ rows from $R$}
	& $R' := \sigma_{\textsc{limit}=x}(R)$
	& $\Dp{R}$
	& $\Cr{R}{a_i}$
	& $\min(x, \Cs{R})$ \\

\cline{2-7}

  & \multirow{3}{*}{\shortstack{Projection\\($\Pi$)}}
	& Standard projection: select attributes $a_i, \ldots$ from $R$
	& $R' := \Pi_{a_i, \ldots}$
	& $\Dp{R}$
	& $\Cr{R}{a_i}$
	& $\Cs{R}$ \\

	\hhline{~|~|---|--}
	&
    & Apply (user-provided, but stateless) $f$ to column $a_i$
	& $R' := \Pi_{f(a_i), \ldots}$
	& $\Dp{R}$
	& $\varnothing$
	& $\Cs{R}$ \\

	\hhline{~|~|---|--}
	&
	& Add range constraint to column $a_i$
	& $R' := \Pi_{a_i \in [l_i,u_i], \ldots}$
	& $\Dp{R}$
	& \makecell[l]{$[l_i,u_i]$ if $a_i \ne \varnothing$ \\$\Cr{R}{a_i}$ otherwise}
	& $\Cs{R}$ \\

\cline{2-7}

 & \multirow{4}{*}{\shortstack{GroupBy\\($\gamma$)}}
    & Group attribute(s) ($g_i$) are \texttt{chunk} (or binned \texttt{chunk}) or \texttt{region}
	& \makecell[l]{$R' := \groupby{g_j,\ldots}{\text{agg}(a_i),\ldots}$ \\ $g_j := \text{chunk} | \text{bin}(\text{chunk})$ }
	& Equation~\ref{eq:base-table-sensitivity}
    & $\D{\text{agg}(a_i)}$
	& $\frac{\Cs{R}}{\text{(bin size)}}$ \\

\cline{3-7}

&
    & Group attribute(s) ($g_j$) are \textit{not} \texttt{chunk} or \texttt{region}
	& $R' := \groupby{g_j,\ldots}{\text{agg}(a_i),\ldots}$
	& $\Dp{R}$
	& $\varnothing$
	& $\varnothing$ \\

&
	& ... discrete set of keys provided for each group (constrains size)
	& $R' := \groupby{g_j \in K_j,\ldots}{\text{agg}(a_i),\ldots}$
	& ...
	& ...
	& $\Pi_j |K_j|$ \\

&
	& ... aggregation constrains range: $agg(a_i) \in [l_i,u_i]$
	& $R' := \groupby{g_j,\ldots}{\text{agg}(a_i) \in [l_i,u_i],\ldots}$
	& ...
	& \makecell[l]{$[l_i,u_i]$ if $a_i \ne \varnothing$\\$\Cr{R}{a_i}$ otherwise}
	& ... \\

\cline{2-7}

  & \multirow{3}{*}{\shortstack{Joins*\\($\Join$)}}
  & *When \textit{immediately} preceeded by GroupBy \textit{over the same key(s)}
	& \multirow{3}{*}{\shortstack[l]{$R' := \groupby{g}{\text{agg}(a)}(R_1 \Join_{g} \ldots \Join_{g} R_n)$\\$R' := \groupby{g}{\text{agg}(a)}(R_1 \fullouterjoin_{g} \ldots \fullouterjoin_{g} R_n)$}}
	& \multirow{3}{*}{$\sum_{i=1}^{n} \Dp{R_i}$}
    & \multirow{3}{*}{\shortstack{(GroupBy\\rules)}}
    & \multirow{3}{*}{\shortstack{(GroupBy\\rules)}} \\

&
	& ... equijoin on $g_j$ (intersection on $g_j$)
&
&
&
& \\

&
	& ... outer join on $g_j$ (union on $g_j$)
&
&
&
& \\

\cline{2-7}
\end{tabular}
\end{center}
\caption{Full set of rules for \system{}'s sensitivity calculation.}
\label{fig:sensitivity-table}
\end{figure*}

\subsection{Proofs}
\label{app:proofs}

\begin{lemma}
    Given a relation $R$, the rules in Figure~\ref{fig:sensitivity-table} are an upper bound on the global sensitivity of a \pkbounded{} event in an intermediate table $t$.
\end{lemma}
\begin{proof}
Proof by induction on the structure of the query.

\noindent\textbf{Case: $t$}. $\Delta_{\policy}(t)$ is given directly by Equation~\ref{eq:base-table-sensitivity}. 

\noindent\textbf{Case: $R' := \sigma_{\theta}(R)$}. A selection may remove some rows from $R$, but it does not add any, or modify any existing ones, so in the worst case an individual can be in just as many rows in $R'$ as in $R$ and thus $\Dp{R'} \le \Dp{R}$ and the constraints remain the same. If $\theta$ includes a $\textsc{limit}=x$ condition, then $R'$ will contain at most $x$ rows, regardless of the number of rows in $R$. 

\noindent\textbf{Case: $R' := \Pi_{a,\ldots}(R)$}. A projection never changes the number of rows, nor does it allow the data in one row to influence another row, so in the worst case an individual can be in at most the same number of rows in $R'$ as in $R$ ($\Dp{R'} \le \Dp{R}$) and the size constraint $\Cs{R}$ remains the same. If the projection transforms an attribute by applying a stateless function $f$ to it, then we can no longer many assumptions about the range of values in $a$ ($\Cr{R'}{a} = \varnothing$), but nothing else changes because the stateless nature of the function ensures that data in row cannot influence any others. 

\noindent\textbf{Case: GroupBy}. A \texttt{GROUP BY} over a fixed set of a $n$ keys is equivalent to $n$ separate queries that use the same aggregation function over a $\sigma_{\textsc{WHERE}col=key}(R)$. If the column being grouped is a user-defined column, \system{} requires that the analyst provide the keys directly. If the column being grouped is one of the two implicit columns (chunk or region), then the set of keys is not dependent on the contents of the data (only its length) and thus are fixed regardless. 

\noindent\textbf{Case: Join}. Consider a query that computes the size of the intersection between two cameras, \texttt{PROCESS}'d into intermediate tables $t_1$ and $t_2$ respectively. If $\Delta(t_1) = x$ and $\Delta(t_2) = y$, it is tempting to assume $\Delta(t_1 \cap t_2) = \min(x,y)$, because a value needs to appear in both $t_1$ and $t_2$ to appear in the intersection. However, because the analyst’s executable can populate the table arbitrarily, they can ``prime'' $t_1$ with values that would only appear in $t_2$, and vice versa. As a result, a value need only appear in either $t_1$ or $t_2$ to show up in the intersection, and thus $\Delta(t_1 \cap t_2) = x + y$ (the sum of the sensitivities of the tables). 

\end{proof}


\begin{theorem}
Consider an adaptive sequence (\SecNS{threat-model}) of $n$ queries $Q_1,\ldots,Q_n$, each over the same camera $C$, a privacy policy $(\rho_C, K_C)$, and global budget $\epsilon_C$. \system{} (Algorithm~\ref{alg:main}) provides $(\rho_C, K_C, \epsilon_C)$-privacy for all $Q_1,\ldots,Q_n$. 

\begin{proof}
Consider two queries $Q_1$ (over time interval $I_1$, using chunk size $c_1$ and budget $\epsilon_1$) and $Q_2$ (over $I_2$, using $c_2$ and $\epsilon_2$). Let $v_1 = V[I_1]$ be the segment of video $Q_1$ analyzes and $v_2 = V[I_2]$ for $Q_2$. 
Let $E$ be a \pkbounded{} event. 

\paragraph{Case 1: $I_1$ and $I_2$ are not $\rho$-disjoint} The budget check (lines 1-3 in Algorithm~\ref{alg:main}) ensures that these two queries must draw from the same privacy budget, because their effective ranges overlap by at least one frame (but may overlap up to all frames). By Theorem~\ref{thm:single-query}, \system{} is $(\rho, K, \epsilon_1)$-private for $Q_1$ and $(\rho, K, \epsilon_2$)-private for $Q_2$. \extendedonly{(Fill in details here).} By Dwork~\cite[Theorem 3.14]{original-dp}, the combination of $Q_1$ and $Q_2$ is $(\rho, K, \epsilon_1 + \epsilon_2)$-private.

\paragraph{Case 2: $I_1$ and $I_2$ are $\rho$-disjoint} In other words, $I_1 + \rho < I_2 - \rho$, thus the budget check (lines 1-3) allows these two queries to draw from entirely separate privacy budgets. Since the intervals are $\rho$-disjoint, and all segments in $E$ must have duration $\le \rho$, it is not possible for the same segment to appear in even a single frame of \emph{both} intervals. 

\paragraph{Case 2a: $E$ is entirely contained within either $I_1$ or $I_2$} Consequently, none of $E$'s segments are contained in the other interval. (Should be straightforward, come back to this, more interesting case is 2b).
\paragraph{Case 2b: $E$ spans some segments in $I_1$ and some in $I_2$}
Let $K_1$ be the number of segments contained in $I_1$, each of duration $\le \rho$, and $K_2$ be the remaining segments contained in $I_2$, each of duration $\le \rho$. In other words, $E$ is $(\rho, K_1)$-bounded in $v_1$ and $(\rho, K_2)$-bounded in $v_2$. Since $E$ has at most $K$ segments, $K_1 + K_2 \le K$. 

We need to show that the probability of observing both $A_1$ and $A_2$ if the inputs are the actual segments $v_1$ and $v_2$ is close ($e^\epsilon$) to the probability of observing those values if the inputs are the neighboring segments $v_1'$ and $v_2'$:
$$ \frac{ \Pr[A_1 = Q_1(v_1), A_2 = Q_2(v_2)] }{ \Pr[A_1 = Q_1(v_1'), A_2 = Q_2(v_2')] } \le \exp(e)$$

\noindent Since the probability of observing $A_1$ is independent of observing $A_2$ (the randomness is purely over the nois added by \system):


\newcommand{\lpdfexp}[2]{\frac{|A_{#1}-Q_{#1}(v_{#1}{#2})|}{b_{#1}}}
\newcommand{\lpdf}[2]{\frac{1}{2b_{#1}}\exp(-\lpdfexp{#1}{#2})}
\newcommand{\pkcsens}[3]{#1(\lceil\frac{#2}{#3}\rceil + 1)}

\begin{align*}
    & \frac{\Pr[A_1 = Q_1(v_1), A_2 = Q_2(v_2)]}{\Pr[A_1 = Q_1(v_1'), A_2 = Q_2(v_2')]} \\
    & \le \frac{\Pr[A_1 = Q_1(v_1)]\Pr[A_2 = Q_2(v_2)]}{\Pr[A_1 = Q_1(v_1')]\Pr[A_2 = Q_2(v_2')]} \\   
    &\le \frac{\lpdf{1}{}\lpdf{2}{}}{\lpdf{1}{'}\lpdf{2}{'}}\\ & \quad\text{(By Algorithm~\ref{alg:main}, Line 13)}\\
    &= {\scriptstyle \exp(\frac{|A_1-Q_1(v_1')| - |A_1-Q_1(v_1)|}{b_1} + \frac{|A_2-Q_2(v_2')| - |A_2-Q_2(v_2)|}{b_2}) }\\
    \shortintertext{If $K_1$ segments are in $v_1$ and $K_2$ segments are in $v_2$, the numerator of each fraction above is the sensitivity of a $(\rho,K_1)$-bounded event and a $(\rho,K_2)$-bounded event, respectively. $b_1$ and $b_2$ are the amount of noise actually added to the query, which are both based on $K$:}
    &\le \exp(\frac{\Delta_{(\rho,K_1)}(Q_1)}{\Delta_{(\rho,K)}(Q_1) / \epsilon} + \frac{\Delta_{(\rho,K_2)}(Q_2)}{\Delta_{(\rho,K)}(Q_2) / \epsilon}) \\
    &= \exp(\epsilon\cdot(\frac{\pkcsens{K_1}{\rho}{c_1}}{\pkcsens{K}{\rho}{c_1}} + \frac{\pkcsens{K_2}{\rho}{c_2}}{\pkcsens{K}{\rho}{c_2}})) \\
    & \quad\text{(by Equation~\ref{eq:base-table-sensitivity})}\\
    &= \exp(\epsilon\cdot(\frac{K_1}{K}+\frac{K_2}{K})) \quad\text{(recall $K \ge K_1+K_2$)}\\
    &\le \exp(\epsilon)
\end{align*}

\end{proof}
\end{theorem}

\extendedonly{
\section{Alternative Privacy Definitions}\label{app:alt-definitions}

Why is \pkeprivacy{} defined exactly the way it is? For example, why not let the privacy policy adjust as the window grows (i.e. protect reptitions of events over longer time periods?) Fundamentally, this does not work, as someone could always issue multiple queries over shorter time periods to extract more information. Prove this more rigorously if possible. Maybe can use the old attempt at proving the old definition as a starting point / counter-example.
}

\section{Masking Optimization}

\tightsubsection{Masking Effectiveness}\label{app:masking-effectiveness}
We extend our evaluation of the potential effectiveness of the masking optimization (\SecNS{practical:mask}) in Table~\ref{tab:masking-effectiveness} by adding 3 videos from BlazeIt~\cite{blazeit} and 4 from MIRIS~\cite{miris}. In each of the 10 total videos we evaluated, there exists a mask that retains a majority of the objects while reducing the maximum persistence by at least an order of magnitude. 

\begin{table*}[]
\small
\centering
\begin{tabular}{|l|l|r|r|r|l|l|}
\hline
\textbf{Dataset}         & \textbf{\begin{tabular}[c]{@{}l@{}}Video \\ Name\end{tabular}} & \multicolumn{1}{l|}{\textbf{\begin{tabular}[c]{@{}l@{}}\% Grid \\ Boxes Masked\end{tabular}}} & \multicolumn{1}{l|}{\textbf{\begin{tabular}[c]{@{}l@{}}Max Perst. \\ before Mask\end{tabular}}} & \multicolumn{1}{l|}{\textbf{\begin{tabular}[c]{@{}l@{}}Max Perst. (frames) \\ after Mask\end{tabular}}} & \textbf{\begin{tabular}[c]{@{}l@{}}Relative Change \\ in Max Perst.\end{tabular}} & \textbf{\begin{tabular}[c]{@{}l@{}}\% Identities Retained \\ After Mask\end{tabular}} \\ \hline
\multirow{3}{*}{Privid}  & \auburn{}                                                      & 17                                                                                            & 1951                                                                                            & 190                                                                                                  & 10.27x                                                                            & 91.06\%                                                                               \\ \cline{2-7} 
                         & \hampton{}                                                     & 30                                                                                             & 28800                                                                                           & 601                                                                                                 & 47.92x                                                                            & 91.3\%                                                                               \\ \cline{2-7} 
                         & \shibuya{}                                                     & 19                                                                                            & 2746                                                                                            & 497.16                                                                                               & 5.52x                                                                             & 87.24\%                                                                               \\ \hline
\multirow{3}{*}{BlazeIt} & grand-canal                                                    & 35                                                                                            & 10930                                                                                           & 2496                                                                                                 & 4.38x                                                                             & 26.67\%                                                                               \\ \cline{2-7} 
                         & venice-rialto                                                  & 6                                                                                             & 37992                                                                                           & 7696                                                                                                 & 4.94x                                                                             & 94.21\%                                                                               \\ \cline{2-7} 
                         & taipei                                                         & 20                                                                                            & 56931                                                                                           & 2444                                                                                                 & 23.29x                                                                            & 99.94\%                                                                               \\ \hline
\multirow{4}{*}{Miris}   & shibuya                                                        & 2                                                                                             & 9363                                                                                            & 2182                                                                                                 & 4.29x                                                                             & 96.43\%                                                                               \\ \cline{2-7} 
                         & beach                                                          & 5                                                                                             & 4843                                                                                            & 843.2                                                                                                & 5.74x                                                                             & 94.79\%                                                                               \\ \cline{2-7} 
                         & warsaw                                                         & 4                                                                                             & 6479                                                                                            & 1147                                                                                                 & 5.65x                                                                             & 94.82\%                                                                               \\ \cline{2-7} 
                         & uav                                                            & 40                                                                                            & 595                                                                                             & 130                                                                                                  & 4.58x                                                                             & 75.57\%                                                                               \\ \hline
\end{tabular}
\caption{Potential effectiveness of masking on videos from an extended dataset, including videos from prior work.}
\label{tab:masking-effectiveness}
\end{table*}

\tightsubsection{Mask to Policy Data Structure}\label{app:mask-data-structure}

%
%
%

\begin{algorithm}
\small
\SetKwInOut{Input}{Input}
\SetKwInOut{Output}{Output}

\Input{ids: set of all detected private bounding boxes}
\Output{something}

boxes\_to\_mask $\leftarrow$ \texttt{[]}
$unmasked\_boxes~\leftarrow~B$\;
\While{$unmasked\_boxes$ is not empty}{
    $max\_track\_id =$ track id with largest persistence\;
    $max\_grid\_box =$ box $b$ intersecting with $max\_track\_id$ for the most number of frames where $b \in unmasked\_boxes$\;
    $boxes\_to\_mask\texttt{.append(}max\_grid\_box\texttt{)}$\;
    \For{track\_id $t$ s.t. $t$ intersects $max\_grid\_box$} {
        \For{frame $f$ s.t. $t$ intersects $max\_grid\_box$ at $f$}{
            remove presence of $t$ at $max\_grid\_box$ for frame $f$\;
            \If {$t$ is no longer present in frame $f$ for all boxes $b$ in $unmasked\_boxes$} {
                persistence of $t \mathrel{-}= 1$\;
            }
        }
    }
    $unmasked\_boxes \mathrel{-}= max\_grid\_box$\;
}
\caption{Generating Ordered List of Boxes to Mask }
\label{alg:gen-masks}
\end{algorithm}

\begin{figure}
  \centering
  \includegraphics[width=\columnwidth]{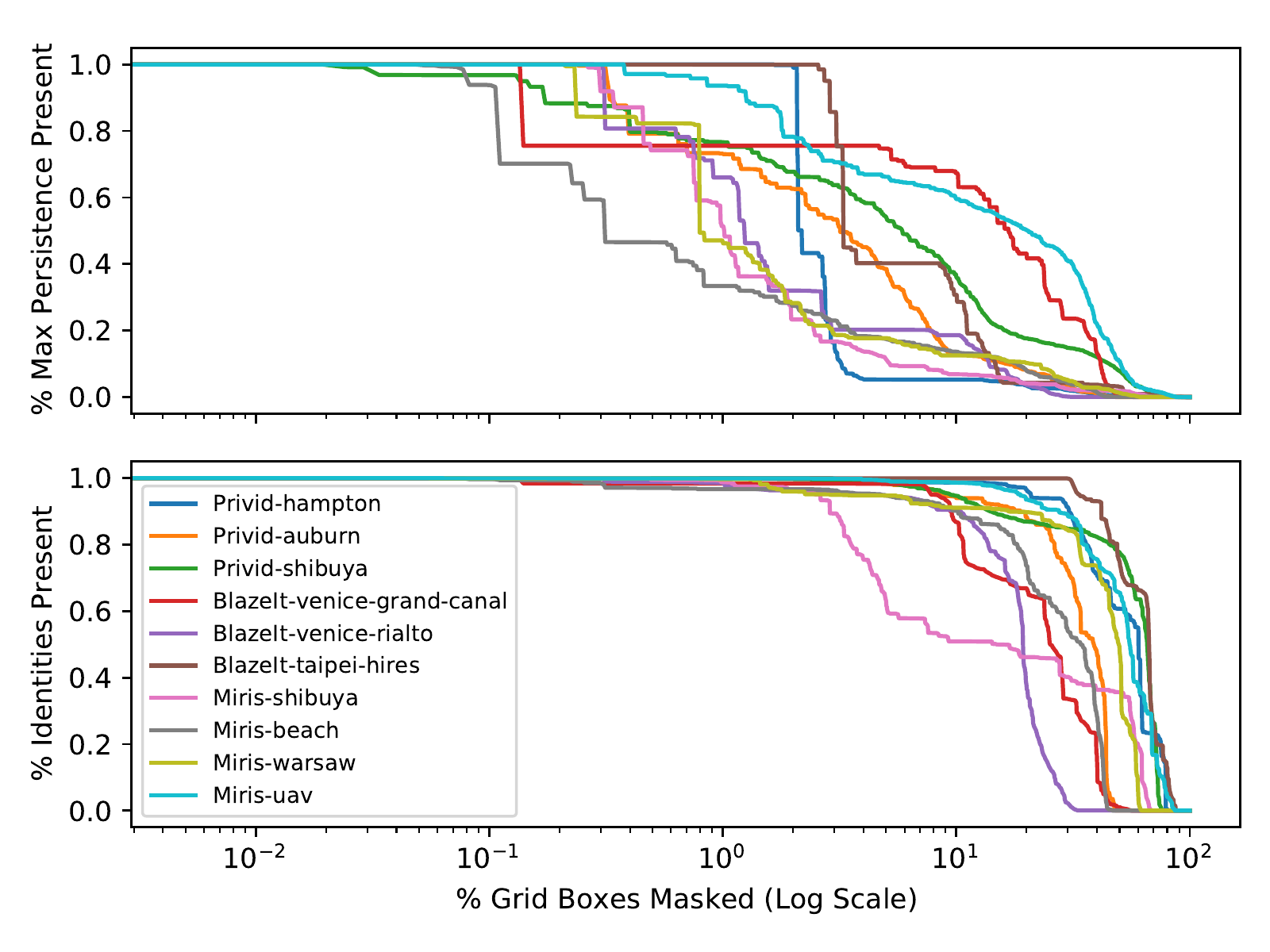}
  \tightcaption{Cumulative Effects of Masking Boxes On Max Persistence and \# Unique Identities Retained. Boxes are masked in the order calculated by Alg.~\ref{alg:gen-masks}. \neil{todo: change names of videos to labels used in privid... change dataset name to p, b, or m}}
  \label{fig:masking-effects}
\end{figure}

In order to decide on the best mask for their specific query, the analyst needs to know how a given mask will impact the persistence, and thus the amount of noise added to their query. 
It is simple for the video owner, with access to the full historical persistence data collected from a camera, to compute the persistence bound after applying a given mask, but this dataset contains individual object tracks of individuals, which is clearly not safe to release to the analyst. Instead, the video owner can compute and releases to the analyst an intermediate data structure, which contains enough information for the analyst to compute the persistence for a given mask without violating privacy.

While there are many possible ways of doing this, we propose one such structure that can be used to choose masks that maximize the persistence decrease with the minimum number of pixels masked. 

A mask $M$ is defined as a set of pixels that will be removed from all frames of the video. 
We start with an empty mask $M_0$ and denote the persistence of the video with this mask to be $\rho_0$. 
The key insight in developing this structure is that, for any given mask (including the initial empty mask) there is a limited set of pixel(s) that one can add to the mask (which we'll call $M_1$ to decrease the overall persistence to $\rho_1$. A tree of masks and associated persistence values, but because of our optimization constraints the tree will be narrow (i.e. typically one child per node, at most a few children), thus tree will be size $\Theta(pixels)$ rather than $\Theta(2^{pixels})$.   

\emph{Claim.} Regardless of the data structure chosen, any structure that maps from potential masks to persistence values does not break our privacy guarantee, because, while it may leak some information, it cannot leak any more information than the analyst already knows about an individual.

\emph{Proof sketch.}
Suppose we have a video in which only a single individual ever appears. In order for us to be confident that a mask lowering persistence indicates the presence of that individual, we have to already know that they are the only individual present, otherwise we don't learn anything. Further, it only tells us that the individual was present in the historical data, but it does not tell us whether or not an individual is present in any particular data we query over, and thus the privacy policy is never violated. 

\Para{Evaluation.} In addition to our heat map visualized in \Fig{heatmap}, we provide an additional way for analysts to understand the effects of masking. We split a frame into a grid of $10x10$ pixel boxes and using the algorithm described in Algorithm~\ref{alg:gen-masks}, generate an ordered list of boxes such the first box decreases the maximum persistence the most, the second boxes decreases it the second most, and so on. As the analyst walks through this list, she can see the resulting cumulative effect on the maximum persistence and the number of unique identities retained. \Fig{masking-effects} plots these relationships, using videos from our dataset as well as videos from BlazeIt~\cite{blazeit} and Miris~\cite{miris}.

\end{appendix}

\end{document}